%% file: Faster.Algo.tex
\documentclass[authoryear, 1p]{elsarticle}
\usepackage{etex}
\usepackage{listings}
\usepackage[utf8]{inputenc} 
\usepackage{fixltx2e} 
\usepackage{mathtools}
\usepackage{wrapfig}
\usepackage{complexity}
\usepackage{listings}
\usepackage{balance}
\usepackage{mathptmx}
\usepackage{mathdots}
\RequirePackage{latexsym}
\RequirePackage{mathtools}
\RequirePackage{amsfonts}
\usepackage{amsthm}
\usepackage{amssymb}
\RequirePackage{array}
\RequirePackage{xspace}
\newcommand{\val}[2]{\texttt{val}^{#1}(#2)}
\newcommand{\suchthat}{\;\ifnum\currentgrouptype=16 \middle\fi|\;}

\input{preamble}
\usetikzlibrary{decorations.pathmorphing}

\journal{arXiv}

\begin{document}

\begin{frontmatter}

\title{Energy Structure of Optimal Positional Strategies \\ in Mean Payoff Games}

\author{Carlo Comin \\ \footnotesize (carlo.comin.86@gmail.com)}

\date{Received: date / Accepted: date}

\begin{abstract}
This note studies structural aspects concerning Optimal Positional Strategies (OPSs) in Mean Payoff Games (MPGs),
it's a contribution to understanding the relationship between OPSs in MPGs
and Small Energy-Progress Measures (SEPMs) in reweighted Energy Games (EGs).
Firstly, it is observed that the space of all OPSs, $\texttt{opt}_{\Gamma}\Sigma^M_0$, admits
a \emph{unique complete decomposition} in terms of so-called \emph{extremal}-SEPM{s} in reweighted EG{s};
this points out what we called the “Energy-Lattice $\mathcal{X}^*_{\Gamma}$ of $\texttt{opt}_{\Gamma}\Sigma^M_0$".
Secondly, it is offered a \emph{pseudo-polynomial total-time} recursive procedure for
\emph{enumerating} (w/o repetitions) all the elements of $\mathcal{X}^*_{\Gamma}$,
and for computing the corresponding partitioning of $\texttt{opt}_{\Gamma}\Sigma^M_0$.
It is observed that the corresponding recursion tree defines an additional lattice $\mathcal{B}^*_{\Gamma}$,
whose elements are certain subgames $\Gamma'\subseteq \Gamma$ that we call \emph{basic} subgames.
The extremal-SEPMs of a given \MPG $\Gamma$ coincide with the least-SEPMs of the basic subgames of $\Gamma$;
so, $\mathcal{X}^*_{\Gamma}$ is the energy-lattice comprising all and only the \emph{least}-SEPMs of
the \emph{basic} subgames of $\Gamma$. The complexity of the proposed enumeration
for both $\mathcal{B}^*_{\Gamma}$ and $\mathcal{X}^*_{\Gamma}$ is $O(|V|^3|E|W |\mathcal{B}^*_{\Gamma}|)$ total time
and $O(|V||E|)+\Theta\big(|E||\mathcal{B}^*_{\Gamma}|\big)$ working space.
Finally, it is constructed an \MPG $\Gamma$ for which $|\mathcal{B}^*_{\Gamma}| > |\mathcal{X}^*_\Gamma|$,
this proves that $\mathcal{B}^*_{\Gamma}$ and $\mathcal{X}^*_\Gamma$ are not isomorphic.
\end{abstract}

\begin{keyword} Mean Payoff Games, Optimal Strategy Synthesis, Pseudo-Polynomial Time, Energy Games, Small Energy-Progress Measures. \end{keyword}

\end{frontmatter}

\input{Faster.Algo-Sect1-Intro}
\input{Faster.Algo-Sect2-Notation-Prelim}
\input{Faster.Algo-Sect4-Values-Optimal-Strategies-Reweightings}
\input{Faster.Algo-Sect6-EnergyDecomp}
\input{Faster.Algo-Sect7-Listing}

\section{Conclusion}\label{sect:conclusions}
We observed a unique complete decomposition of $\texttt{opt}_{\Gamma}\Sigma^M_0$ in terms of extremal-SEPM{s} in reweighted \EG{s},
also offering a pseudo-polynomial total-time recursive algorithm for enumerating (w/o repetitions) all the elements of $\mathcal{X}^*_\Gamma$,
\ie all extremal-SEPMs, and for computing the components of the corresponding partitioning $\mathcal{B}^*_{\Gamma}$ of $\texttt{opt}_{\Gamma}\Sigma^M_0$.

It would be interesting to study further properties enjoyed by $\mathcal{B}^*_{\Gamma}$ and $\mathcal{X}^*_\Gamma$;
and we ask for more efficient algorithms for enumerating $\mathcal{X}^*_\Gamma$,
\eg pseudo-polynomial time-delay and \emph{polynomial space} enumerations.

\paragraph*{Acknowledgments}
This work was supported by \emph{Department of Computer Science, University of Verona, Verona, Italy},
under PhD grant “Computational Mathematics and Biology”, on a co-tutelle agreement with
\emph{LIGM, Universit\'e Paris-Est in Marne-la-Vall\'ee, Paris, France}.

\bibliographystyle{elsarticle-num-names}
\bibliography{biblio}


\end{document}

%% file: preamble.tex
\newcommand{\ie}{i.e., }
\newcommand{\eg}{e.g., }
\newcommand{\etal}{\textit{et al.}\xspace}

\newcommand{\wrt}{w.r.t.\ }

\usepackage[textsize=small]{todonotes}

\newcommand{\Q}{\mathbb{Q}\xspace}
\newcommand{\N}{\mathbb{N}\xspace}
\newcommand{\Z}{\mathbb{Z}\xspace}

\newcommand{\MPG}{MPG\xspace}
\newcommand{\EG}{EG\xspace}
\newcommand{\MCP}{MCP\xspace}

\def\C{{\cal C}}

\def\W{{\cal W}}

\newcommand{\figref}[1]{Fig.~\ref{#1}}

\newtheorem{Thm}{Theorem}
\newtheorem{Exp}{Example}
\newtheorem{Lem}{Lemma}
\newtheorem{Prop}{Proposition}

\newtheorem{Def}{Definition}
\newtheorem{Rem}{Remark}

\usepackage[ruled,vlined,linesnumbered]{algorithm2e}
\newcounter{proccnt}
\newenvironment{algo-proc}[1][htb]
  {\refstepcounter{proccnt}\renewcommand*{\algorithmcfname}{SubProcedure} 
   \begin{algorithm}[#1]\renewcommand\thealgocf{\arabic{proccnt}}
}{\end{algorithm}\addtocounter{algocf}{-1}}
\usepackage{subfig}
\SetCommentSty{textsf}
\SetKwRepeat{DoWhile}{do}{while}
\SetAlFnt{\small} 

\makeatletter
\newcommand{\removelatexerror}{\let\@latex@error\@gobble}
\makeatother

\let\oldnl\nl
\newcommand{\nonl}{\renewcommand{\nl}{\let\nl\oldnl}}
\usepackage{tikz}
\usepackage{tikz-qtree}
\usepackage{pgfplots}
\usetikzlibrary{pgfplots.units}
\usetikzlibrary{calc,positioning,fit}
\usetikzlibrary{shapes,shapes.multipart,shapes.arrows}
\usetikzlibrary{decorations,decorations.pathmorphing,decorations.pathreplacing,decorations.markings,decorations.shapes}
\usetikzlibrary{arrows}
\usetikzlibrary{fit,backgrounds}
\usetikzlibrary{plotmarks}
\tikzstyle{node}=[circle,draw,inner sep=2pt,transform shape,minimum size=1.75em]
 
\newcommand*\sizedcircled[2]{\tikz[baseline=(char.base)]{ \node[shape=circle,draw,inner sep=2pt, scale=#1] (char) {\textbf{#2}}; }} 

\tikzstyle{every picture}=[>=latex]
\tikzstyle{every label}=[inner sep=2pt]

%% file: Faster.Algo-Sect1-Intro.tex
\section{Introduction}\label{sect:introduction}
A \emph{Mean Payoff Game} (\MPG) is a two-player infinite game $\Gamma=(V, E, w, \langle V_0, V_1 \rangle)$,
that is played on a finite weighted directed graph, denoted $G^{\Gamma} \triangleq ( V, E, w )$, where $w:E\rightarrow \Z$,
the vertices of which are partitioned into two classes, $V_0$ and $V_1$, according to the player to which they belong.

At the beginning of the game a pebble is placed on some vertex $v_s\in V$, then the two players,
named Player~0 and Player~1, move it along the arcs ad infinitum. Assuming the pebble is currently on some $v\in V_0$,
then Player~0 chooses an arc $(v,v')\in E$ going out of $v$ and moves the pebble to the destination vertex $v'$.
Similarly, if the pebble is currently on some $v\in V_1$, it is Player~1's turn to choose an outgoing arc.
The infinite sequence $v_s,v,v'\ldots$ of all the encountered vertices forms a \emph{play}.
In order to play well, Player~0 wants to maximize the limit inferior of the long-run average weight
of the traversed arcs, \ie to maximize $\liminf_{n\rightarrow\infty} \frac{1}{n}\sum_{i=0}^{n-1}w(v_i, v_{i+1})$, whereas
Player~1 wants to minimize the $\limsup_{n\rightarrow\infty} \frac{1}{n} \sum_{i=0}^{n-1}w(v_i, v_{i+1})$.
\cite{EhrenfeuchtMycielski:1979}~proved that each vertex $v$ admits a \emph{value}, denoted $\val{\Gamma}{v}$,
that each player can secure by means of a \emph{memoryless} (or \emph{positional}) strategy,
\ie one depending only on the current vertex position and not on the previous choices.

Solving an \MPG consists in computing the values of all vertices (\emph{Value Problem}) and, for each player,
a positional strategy that secures such values to that player (\emph{Optimal Strategy Synthesis}).
The corresponding decision problem lies in $\NP\cap \coNP$~\citep{ZwickPaterson:1996} and it was later
shown to be in $\UP\cap\coUP$~\citep{Jurdzinski1998}.

The problem of devising efficient algorithms for solving \MPG{s} has been studied extensively in the literature.
The first milestone was settled in \cite{GKK88}, in which it was offered an \emph{exponential} time algorithm
for solving a slightly wider class of \MPG{s} called \emph{Cyclic Games}.
Afterwards, \cite{ZwickPaterson:1996} devised the first deterministic procedure for computing values in \MPG{s},
and optimal strategies securing them, within a pseudo-polynomial time and polynomial space.
In particular, it was established an $O(|V|^3 |E| W)$ upper bound for the time complexity of the Value Problem,
as well as an upper bound of $O(|V|^4|E| W \log(|E|/|V|))$ for that of Optimal Strategy Synthesis~\citep{ZwickPaterson:1996}.

Several research efforts have been spent in studying quantitative extensions of infinite games for
	modeling quantitative aspects of reactive systems, \eg the \emph{Energy Games (EGs)}~\citep{Chakrabarti03,Bouyer08,brim2011faster}.
These studies unveiled interesting connections between EGs and \MPG{s}; and by relying on these techniques,
	recently the worst-cast time complexity of the Value Problem and Optimal Strategy Synthesis was given an
	improved pseudo-polynomial upper bound~\citep{CR15, CR16}; those works focused on offering a simple proof of the improved upper bound.
However, the running time of the proposed algorithm turned out to be also $\Omega(|V|^2|E|W)$, the actual time complexity being
$\Theta\big(|V|^2 |E|\, W + \sum_{v\in V}\texttt{deg}_{\Gamma}(v)\cdot\ell_{\Gamma}^0(v)\big)$, where
	$\ell_{\Gamma}^0(v)\leq (|V|-1)|V|W$ denotes the total number of times that a certain energy-lifting operator $\delta(\cdot, v)$ is applied to any $v\in V$.
A way to overcome this issue was found in~\cite{CominR16a}, where a novel algorithmic scheme, named \emph{Jumping}, was introduced;
	by tackling on some further regularities of the problem, the estimate on the pseudo-polynomial time complexity of \MPG{s} was reduced to:
		$O(|E|\log |V|) + \Theta\big(\sum_{v\in V}\texttt{deg}_{\Gamma}(v)\cdot\ell_{\Gamma}^1(v)\big)$,
	where, for every $v\in V$, $\ell_{\Gamma}^1(v)$ is the total number of applications of $\delta(\cdot, v)$ that are made by the algorithm;
	$\ell_{\Gamma}^1\leq (|V|-1)|V|W$ (worst-case, but experimentally $\ell_{\Gamma}^1\ll\ell_{\Gamma}^0$;
	see~\cite{CominR16a}),	and the working space is $\Theta(|V|+|E|)$.
With this, the pseudo-polynomiality was confined to depend solely on the total number $\ell^1_{\Gamma}$ of required energy-liftings.

\paragraph*{Contribution}
This work studies the relationship between Optimal Positional Strategies (OPSs) in MPGs
	and Small Energy-Progress Measures (SEPMs) in reweighted EGs.
Actually this paper is an extended and revised version of Section~5 in \cite{CR15}.
Here, we offer:
\begin{itemize}
 	\item[1.] \emph{An Energy-Lattice Decomposition of the Space of Optimal Positional Strategies in MPGs.}
\end{itemize}
Let's denote by $\texttt{opt}_{\Gamma}\Sigma^M_0$ the space of all the optimal positional strategies in a given \MPG $\Gamma$.
What allows the algorithms given in~\cite{CR15, CR16, CominR16a} to compute at least
one $\sigma^*_0\in \texttt{opt}_{\Gamma}\Sigma^M_0$ is a \emph{compatibility} relation
that links optimal arcs in \MPG{s} to arcs that are \emph{compatible} \wrt least-SEPM{s} in reweighted \EG{s}.
The family $\mathcal{E}_\Gamma$ of all SEPM{s} of a given \EG $\Gamma$ forms a complete finite lattice, the Energy-Lattice of the \EG $\Gamma$.
Firstly, we observe that even though compatibility \wrt \emph{least}-SEPMs in reweighted \EG{s} implies optimality of positional strategies
in \MPG{s} (see Theorem~\ref{Thm:pos_opt_strategy}), the converse doesn't hold generally (see Proposition~\ref{prop:counter_example}).
Thus a natural question was whether compatibility \wrt SEPM{s} was really appropriate to capture (\eg to provide a recursive enumeration of) the
whole $\texttt{opt}_{\Gamma}\Sigma^M_0$ and not just a proper subset of it.
Partially motivated by this question we explored on the relationship between $\texttt{opt}_{\Gamma}\Sigma^M_0$
and $\mathcal{E}_\Gamma$. In Theorem~\ref{thm:main_energystructure}, it is observed a unique complete decomposition of
$\texttt{opt}_{\Gamma}\Sigma^M_0$ which is expressed in terms of so called \emph{extremal}-SEPM{s} in reweighted EG{s}.
This points out what we called the “Energy-Lattice $\mathcal{X}^*_{\Gamma}$ associated to $\texttt{opt}_{\Gamma}\Sigma^M_0$",
the family of all the extremal-SEPM{s} of a given \MPG $\Gamma$.
So, compatibility \wrt SEPM{s} actually turns out to be appropriate for constructing the whole $\texttt{opt}_{\Gamma}\Sigma^M_0$;
but an entire lattice $\mathcal{X}^*_{\Gamma}$ of extremal-SEPM{s} then arises (and not just the least-SEPM,
which turns out to account only for the join/top component of $\texttt{opt}_{\Gamma}\Sigma^M_0$).

\begin{itemize}
	\item[2.] \emph{A Recursive Enumeration of Extremal-SEPMs and Optimal Positional Strategies in MPGs.}
\end{itemize}
It is offered a pseudo-polynomial total time recursive procedure
for enumerating (w/o repetitions) all the elements of $\mathcal{X}^*_{\Gamma}$,
and for computing the associated partitioning of $\texttt{opt}_{\Gamma}\Sigma^M_0$.
This shows that the above mentioned compatibility relation is appropriate so to extend the algorithm given in~\cite{CominR16a},
recursively, in order to compute the whole $\texttt{opt}_{\Gamma}\Sigma^M_0$ and $\mathcal{X}^*_{\Gamma}$.
It is observed that the corresponding recursion tree actually defines an additional lattice $\mathcal{B}^*_{\Gamma}$,
whose elements are certain subgames $\Gamma'\subseteq \Gamma$ that we call \emph{basic} subgames.
The extremal-SEPMs of a given $\Gamma$ coincide with the least-SEPMs of the basic subgames of $\Gamma$;
so, $\mathcal{X}^*_{\Gamma}$ is the energy-lattice comprising all and only the \emph{least}-SEPMs of
the \emph{basic} subgames of $\Gamma$. The total time complexity of the proposed enumeration
for both $\mathcal{B}^*_{\Gamma}$ and $\mathcal{X}^*_{\Gamma}$ is $O(|V|^3|E|W |\mathcal{B}^*_{\Gamma}|)$,
it works in space $O(|V||E|)+\Theta\big(|E||\mathcal{B}^*_{\Gamma}|\big)$.
An example of \MPG $\Gamma$ for which $|\mathcal{B}^*_{\Gamma}| > |\mathcal{X}^*_\Gamma|$ ends this paper.

\paragraph*{Organization}
The following Section~\ref{sect:background} introduces some notation and provides the
required background on infinite 2-player pebble games and related algorithmic results.
In Section~\ref{section:values}, a suitable relation between values, optimal strategies,
and certain reweighting operations is recalled from~\cite{CR15, CR16}.
Section~\ref{sect:energy} offers a unique and complete energy-lattice decomposition of $\texttt{opt}_{\Gamma}\Sigma^M_0$.
Finally, Section~\ref{sect:recursive_enumeration} provides a recursive enumeration
of $\mathcal{X}^*_{\Gamma}$ and the corresponding partitioning of $\texttt{opt}_{\Gamma}\Sigma^M_0$.

%% file: Faster.Algo-Sect2-Notation-Prelim.tex
\section{Notation and Preliminaries}\label{sect:background}
We denote by $\N$, $\Z$, $\Q$ the set of natural, integer, and rational numbers.
It will be sufficient to consider integral intervals,
  \eg $[a,b]\triangleq\{z\in\Z\mid a\leq z\leq b\}$
  and $[a,b)\triangleq\{z\in\Z\mid a \leq z < b\}$ for any $a,b\in \Z$.
Our graphs are directed and weighted on the arcs; thus, if $G=(V, E, w)$ is a graph,
  then every arc $e\in E$ is a triplet $e=(u,v,w_e)$, where $w_e = w(u,v)\in\Z$.
Let $W \triangleq \max_{e\in E} |w_e|$ be the maximum absolute weight.
Given a vertex $u\in V$, the set of its successors is $N_\Gamma^{\text{out}}(u) \triangleq \{ v \in V \mid (u,v) \in E \}$,
  and the set of its predecessors is $N_\Gamma^{\text{in}}(u) \triangleq \{ v \in V \mid (v,u) \in E \}$.
Let $\texttt{deg}_\Gamma(v)\triangleq |N_\Gamma^{\text{in}}(v)| + |N_\Gamma^{\text{out}}(v)|$.
A \emph{path} is a sequence $v_0v_1\ldots v_n\ldots$ such that $\forall^{i\in [n]}\, (v_i, v_{i+1}) \in E$.
Let $V^*$ be the set of all (possibly empty) finite paths.
A \emph{simple path} is a finite path $v_0v_1\ldots v_n$ having no repetitions,
  \ie for any $i,j\in [0,n]$ it holds $v_i \neq v_j$ if $i\neq j$.
A \emph{cycle} is a path $v_0v_1\ldots v_{n-1}v_n$ such that $v_0\ldots v_{n-1}$ is simple and $v_n = v_0$.
The \emph{average weight} of a cycle $v_0\ldots v_n$ is $w(C)/|C|=\frac{1}{n} \sum_{i=0}^{n-1} w(v_i,v_{i+1})$.
A cycle $C=v_0v_1\ldots v_n$ is \emph{reachable} from $v$ in $G$ if
  there is some path $p$ in $G$ such that $p\cap C\neq\emptyset$.

An \emph{arena} is a tuple $\Gamma = (V, E, w, \langle V_0, V_1\rangle)$
where $G^{\Gamma} \triangleq (V, E, w)$ is a finite weighted directed graph and $(V_0, V_1)$ is a partition
of $V$ into the set $V_0$ of vertices owned by Player~0, and $V_1$ owned by Player~$1$.
It is assumed that $G^{\Gamma}$ has no sink, \ie $\forall^{v\in V} N_{\Gamma}^{\text{out}}(v)\neq\emptyset$;
we remark that $G^{\Gamma}$ is not required to be a bipartite graph on colour classes $V_0$ and $V_1$.
A subarena $\Gamma'$ (or \emph{subgame}) of $\Gamma$ is any arena $\Gamma' = (V', E', w', \langle V'_0, V'_1\rangle )$
such that: $V'\subseteq V$, $\forall^{i\in\{0,1\}} V'_i=V'\cap V_i$, $E'\subseteq E$, and $\forall^{e\in E'} w'_e=w_e$.
Given $S\subseteq V$, the subarena of $\Gamma$ induced by $S$ is denoted $\Gamma_{|_{S}}$,
its vertex set is $S$ and its edge set is $E'=\{(u,v)\in E \mid u,v\in S\}$.
A game on $\Gamma$ is played for infinitely many rounds by two players moving a pebble along the arcs
of $G^{\Gamma}$. At the beginning of the game the pebble is found on some vertex $v_s\in V$,
which is called the \emph{starting position} of the game. At each turn,
assuming the pebble is currently on a vertex $v\in V_i$ (for $i=0, 1$),
Player~$i$ chooses an arc $(v,v')\in E$ and then the next turn starts with the pebble on $v'$.
Below, \figref{fig:ex1_arena} depicts an example arena $\Gamma_{\text{ex}}$.

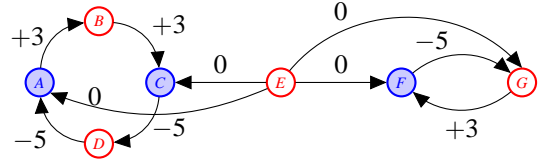
\begin{figure}[!h]
\center
\begin{tikzpicture}[scale=.6,arrows={-triangle 45}, node distance=1.5 and 2]
 		\node[node, thick, color=red] (E) {$E$};
		\node[node, thick, color=blue, left=of E, fill=blue!20] (C) {$C$};
		\node[node, thick, color=red, above=of C, xshift=-8.5ex, yshift=-5ex] (B) {$B$};
		\node[node, thick, color=blue, left=of C, fill=blue!20] (A) {$A$};
		\node[node, thick, color=red, below=of C, xshift=-8.5ex, yshift=5ex] (D) {$D$};
		\node[node, thick, color=blue, right=of E, fill=blue!20] (F) {$F$};
		\node[node, thick, color=red, right=of F] (G) {$G$};
		\draw[] (E) to [bend left=0] node[above] {$0$} (C);
		\draw[] (E) to [bend left=22] node[above left, xshift=-4ex] {$0$} (A.south east);
		\draw[] (E) to [bend left=50] node[above left, xshift=-4ex, yshift=-1ex] {$0$} (G.north);
		\draw[] (E) to [bend left=0] node[above] {$0$} (F);
		\draw[] (A) to [bend left=40] node[left] {$+3$} (B);
		\draw[] (B) to [bend left=40] node[xshift=2ex, yshift=1ex] {$+3$} (C);
		\draw[] (C) to [bend left=40] node[xshift=2ex, yshift=0ex] {$-5$} (D);
		\draw[] (D) to [bend left=40] node[xshift=-2ex, yshift=-1ex] {$-5$} (A);
		\draw[] (F) to [bend left=40] node[xshift=-2.5ex, yshift=-.5ex, above] {$-5$} (G);
		\draw[] (G) to [bend left=40] node[below] {$+3$} (F);
\end{tikzpicture}
\caption{
An arena $\Gamma_{\text{ex}}=\langle V, \E, w, (V_0, V_1) \rangle$. Here, $V=\{A,B,C,D,E,F,G\}$ and
$\E=\{(A,B,+3), (B,C,+3), (C,D,-5)$, $(D,A,-5), (E,A,0), (E,C,0), (E,F,0), (E,G,0), (F,G,-5), (G,F,+3)\}$.
Also, $V_0=\{B,D,E,G\}$ is colored in red, while $V_1=\{A,C,F\}$ is filled in blue.
}\label{fig:ex1_arena}
\end{figure}

A \emph{play} is any infinite path $v_0v_1\ldots v_n\ldots\in V^\omega$ in $\Gamma$.
For any $i\in \{0,1\}$, a strategy of Player~$i$ is any function $\sigma_i:V^*\times V_i\rightarrow V$
such that for every finite path $p'v$ in $G^{\Gamma}$,
where $p'\in V^*$ and $v\in V_i$, it holds that $(v, \sigma_i(p', v))\in E$.
A strategy $\sigma_i$ of Player $i$ is \emph{positional} (or \emph{memoryless})
if $\sigma_i(p, v_n) = \sigma_i(p', v'_m)$ for every finite paths $pv_n=v_0\ldots v_{n-1}v_n$
and $p'v'_m=v'_0\ldots v'_{m-1}v'_m$ in $G^{\Gamma}$ such that $v_n=v'_m\in V_i$.
The set of all the positional strategies of Player~$i$ is denoted by $\Sigma^M_i$.
A play $v_0v_1\ldots v_n\ldots $ is \emph{consistent} with a strategy
$\sigma\in\Sigma_i$ if $v_{j+1} = \sigma(v_0v_1\ldots v_j)$ whenever $v_j\in V_i$.

Given a starting position $v_s\in V$, the \emph{outcome} of two strategies $\sigma_0 \in\Sigma_0$ and $\sigma_1 \in\Sigma_1$,
denoted $\texttt{outcome}^{\Gamma}(v_s, \sigma_0, \sigma_1)$,
is the unique play that starts at $v_s$ and is consistent with both $\sigma_0$ and $\sigma_1$.

Given a memoryless strategy $\sigma_i\in\Sigma^M_i$ of Player~$i$ in $\Gamma$,
then $G(\sigma_i, \Gamma)=(V, E_{\sigma_i}, w)$ is the graph obtained from $G^{\Gamma}$
by removing all the arcs $( v,v')\in E$ such that $v\in V_i$ and $v'\neq \sigma_i(v)$;
we say that $G(\sigma_i, \Gamma)$ is obtained from $G^{\Gamma}$ \emph{by projection} \wrt $\sigma_i$.

For any weight function $w':E\rightarrow \Z $, the \emph{reweighting} of $\Gamma=(V, E, w, \langle V_0, V_1\rangle )$ \wrt $w'$
is the arena $\Gamma^{w'}= (V, E, w', \langle V_0, V_1\rangle )$. Also, for $w:E\rightarrow \Z$ and any $\nu\in \Z$,
we denote by $w+\nu$ the weight function $w'$ defined as $\forall^{e\in E} w'_e \triangleq w_e+\nu$.
Indeed, we shall consider reweighted games of the form $\Gamma^{w-q}$, for some $q\in \Q$.
Notice that the corresponding weight function $w':E\rightarrow\Q:e\mapsto w_e-q$ is rational,
while we required the weights of the arcs to be always integers.
To overcome this issue, it is sufficient to re-define $\Gamma^{w-q}$ by  scaling all weights by a factor equal
to the denominator of $q\in \Q$; \ie when $q\in \Q$,
say $q=N/D$ for $\gcd(N,D)=1$ we define $\Gamma^{w-q}\triangleq \Gamma^{D\cdot w-N}$.
This rescaling operation doesn't change the winning regions of the corresponding games,
let's denote this equivalence as $\Gamma^{w-q}\cong \Gamma^{D\cdot w - N}$,
and it has the significant advantage of allowing for a discussion (and an algorithmics) which is strictly based on integer weights.

\subsection{Mean Payoff Games}
A \emph{Mean Payoff Game}~(\MPG)~\citep{brim2011faster, ZwickPaterson:1996, EhrenfeuchtMycielski:1979}
is a game played on some arena $\Gamma$ for infinitely many rounds by two opponents,
Player~$0$ gains a payoff defined as the long-run average weight of the play,
whereas Player~$1$ loses that value.
Formally, the Player~$0$'s \emph{payoff} of a play $v_0v_1\ldots v_n\ldots $
in $\Gamma$ is defined as follows:
\[\texttt{MP}_0(v_0v_1\ldots v_n\ldots)\triangleq\liminf_{n\rightarrow\infty}
	\frac{1}{n}\sum_{i=0}^{n-1}w(v_i, v_{i+1}).\]
The value \emph{secured} by a strategy $\sigma_0\in\Sigma_0$ in a vertex $v$ is defined as:
\[\texttt{val}^{\sigma_0}(v)\triangleq
\inf_{\sigma_1\in\Sigma_1}\texttt{MP}_0\big(\texttt{outcome}^{\Gamma}(v, \sigma_0, \sigma_1)\big),\]
Notice that payoffs and secured values can be defined symmetrically for the Player~$1$
(\ie by interchanging the symbol \emph{0} with \emph{1} and \emph{inf} with \emph{sup}).

Ehrenfeucht and Mycielski~\cite{EhrenfeuchtMycielski:1979} proved that each vertex
$v\in V$ admits a unique \emph{value}, denoted $\val{\Gamma}{v}$, which each player can secure by means
of a \emph{memoryless} (or \emph{positional}) strategy. Moreover,
\emph{uniform} positional optimal strategies do exist for both players,
in the sense that for each player there exist at least one
positional strategy which can be used to secure all the optimal values,
independently with respect to the starting position $v_s$.
Thus, for every \MPG $\Gamma$, there exists a strategy
$\sigma_0\in\Sigma^M_0$ such that $\forall^{v\in V} \val{\sigma_0}{v}\geq \val{\Gamma}{v}$,
and there exists a strategy $\sigma_1\in\Sigma^M_1$ such that $\forall^{v\in V} \val{\sigma_1}{v}\leq \val{\Gamma}{v}$.
The \emph{(optimal) value} of a vertex $v\in V$ in the \MPG $\Gamma$ is given by:
\[\val{\Gamma}{v} = \sup_{\sigma_0\in\Sigma_0}\val{\sigma_0}{v} = \inf_{\sigma_1\in\Sigma_1}\val{\sigma_1}{v}.\]
Thus, a strategy $\sigma_0\in\Sigma_0$ is \emph{optimal} if $\texttt{val}^{\sigma_0}(v)=\val{\Gamma}{v}$ for all $v\in V$.
We denote $\text{opt}_{\Gamma}\Sigma^M_0\triangleq \big\{\sigma_0\in\Sigma^M_0(\Gamma) \mid \;
	\forall^{v\in V}\, \texttt{val}^{\Gamma}_{\sigma_0}(v) = \val{\Gamma}{v}\big\}$.
A strategy $\sigma_0\in\Sigma_0$ is said to be \emph{winning} for Player~$0$ if $\forall^{v\in V}\texttt{val}^{\sigma_0}(v)\geq 0$,
and $\sigma_1\in\Sigma_1$ is winning for Player~$1$ if $\texttt{val}^{\sigma_1}(v) < 0$.
Correspondingly, a vertex $v\in V$ is a \emph{winning starting position} for
Player~$0$ if $\val{\Gamma}{v}\geq 0$, otherwise it is winning for Player~$1$.
The set of all winning starting positions of Player~$i$ is denoted by $\W_i$ for $i\in \{0,1\}$.

A refined formulation of the determinacy theorem is offered in~\cite{Bjorklund04}.
\begin{Thm}[\cite{Bjorklund04}]\label{thm:ergodic_partition}
Let $\Gamma$ be an \MPG and let $\{C_i\}_{i=1}^m$ be a partition
(called \emph{ergodic}) of its vertices into $m\geq 1$ classes each one having the same optimal value $\nu_i\in\Q$.
Formally, $V=\bigsqcup_{i=1}^m C_i$ and $\forall^{i\in [m]}\forall^{v\in C_i} \val{\Gamma_{i}}{v}=\nu_i$,
where $\Gamma_{i}\triangleq\Gamma_{|_{C_i}}$.

Then, Player~0 has no vertices with outgoing arcs leading from $C_i$ to $C_j$ whenever $\nu_i<\nu_j$,
	and Player~1 has no vertices with outgoing arcs leading from $C_i$ to $C_j$ whenever $\nu_i>\nu_j$;

moreover, there exist $\sigma_0\in\Sigma^M_0$ and $\sigma_1\in\Sigma^M_1$ such that:

-- If the game starts from any vertex in $C_i$,
	then $\sigma_0$ secures a gain at least $\nu_i$ to Player~0 and $\sigma_1$ secures a loss at most $\nu_i$ to Player~1;

-- Any play that starts from $C_i$ always stays in $C_i$, if it is consistent with both strategies $\sigma_0, \sigma_1$,
	\ie if Player~0 plays according to $\sigma_0$, and Player~1 according to $\sigma_1$.
\end{Thm}

A finite variant of \MPG{s} is well-known in the literature
\citep{EhrenfeuchtMycielski:1979, ZwickPaterson:1996, brim2011faster},
where the game stops as soon as a cyclic sequence of vertices is traversed.
It turns out that this is equivalent to the infinite game formulation~\citep{EhrenfeuchtMycielski:1979},
in the sense that the values of an \MPG are in a strong relationship
with the average weights of its cycles, as in the next lemma.

\begin{Prop}[Brim,~\etal \cite{brim2011faster}]\label{prop:reachable_cycle}
Let $\Gamma$ be an \MPG. For all $\nu\in\Q$, for all $\sigma_0\in\Sigma^M_0$,
and for all $v\in V$, the value $\texttt{val}^{\sigma_0}(v)$ is greater than $\nu$ \textit{iff} all cycles $C$
reachable from $v$ in the projection graph $G^{\Gamma}_{\sigma_0}$ have an average weight $w(C)/|C|$ greater than $\nu$.
\end{Prop}
The proof of Proposition~\ref{prop:reachable_cycle} follows from the memoryless determinacy of \MPG{s}.
We remark that a proposition which is symmetric to Proposition~\ref{prop:reachable_cycle} holds for Player~$1$ as well:
for all $\nu\in\Q$, for all positional strategies $\sigma_1\in\Sigma^M_1$ of Player~$1$, and for all vertices $v\in V$,
the value $\texttt{val}^{\sigma_1}(v)$ is less than $\nu$ \textit{iff} if
all cycles reachable from $v$ in the projection graph $G^{\Gamma}_{\sigma_1}$
have an average weight less than $\nu$. Also,
it is well-known~\citep{brim2011faster, EhrenfeuchtMycielski:1979} that each value $\val{\Gamma}{v}$ is contained within
the following set of rational numbers: \[ S_{\Gamma}=\Big\{ N/D \suchthat D\in [1, |V|],\, N\in [-D\cdot W, D\cdot W] \Big\}.\]
Notice, $|S_{\Gamma}|\leq |V|^2 W$.

The present work focuses on the algorithmics of the following classical problem:

-- \emph{Optimal Strategy Synthesis.} Compute an optimal positional strategy for Player~0 in $\Gamma$.

Also, in Section~\ref{sect:recursive_enumeration} we shall consider the problem of computing the whole $\texttt{opt}_\Gamma\Sigma^M_0$:

-- \emph{Optimal Strategy Enumeration.} Provide a listing\footnote{The listing has to be exhaustive
  (\ie each element is listed eventually) and without repetitions (\ie no element is listed twice).}
  of all the optimal positional strategies of Player~0 in the \MPG $\Gamma$.

\subsection{Energy Games and Small Energy-Progress Measures}
An \emph{Energy Game} (\EG) is a game that is played on an arena $\Gamma$ for infinitely many rounds by two opponents,
where the goal of Player~0 is to construct an infinite play $v_0v_1\ldots v_n\ldots$
such that for some initial \emph{credit} $c\in\N$ the following holds:
$c + \sum_{i=0}^{j}w(v_i, v_{i+1})\geq 0\, \text{, for all } j \geq 0$.
Given an initial credit $c\in\N$, a play $v_0v_1\ldots v_n\ldots$ is \emph{winning} for Player~0 if it satisfies (1),
otherwise it is winning for Player~1. A vertex $v\in V$ is a winning
starting position for Player~0 if there exists an initial credit $c\in\N$
and a strategy $\sigma_0\in\Sigma_0$ such that, for every strategy $\sigma_1\in\Sigma_1$,
the play $\texttt{outcome}^{\Gamma}(v, \sigma_0, \sigma_1)$ is winning for Player~0.
As in the case of \MPG{s}, the \EG{s} are memoryless determined \cite{brim2011faster},
\ie for every $v\in V$, either $v$ is winning for Player~$0$ or $v$ is winning for Player~$1$,
and (uniform) memoryless strategies are sufficient to win the game.
In fact, as shown in the next lemma, the decision problems of \MPG{s} and \EG{s} are intimately related.
\begin{Prop}[\cite{brim2011faster}]\label{prop:relation_MPG_EG}
Let $\Gamma$ be an arena. For all threshold $\nu\in\Q$, for all vertices $v\in V$,
Player~$0$ has a strategy in the \MPG $\Gamma$ that secures value at least $\nu$ from $v$ if and only if,
for some initial credit $c\in\N$, Player~$0$ has a winning strategy from $v$ in the reweighted \EG $\Gamma^{w-\nu}$.
\end{Prop}

In this work we are especially interested in the \emph{Minimum Credit Problem} (\MCP) for \EG{s}:
for each winning starting position $v$, compute the minimum initial credit $c^*=c^*(v)$
such that there exists a winning strategy $\sigma_0\in\Sigma^M_0$ for Player~$0$ starting from $v$.
A fast pseudo-polynomial time deterministic procedure for solving \MCP{s} comes from \cite{brim2011faster}.

\begin{Thm}[\cite{brim2011faster}]\label{Thm:VI}
There exists a deterministic algorithm for solving the MCP within $O(|V|\,|E|\,W)$ pseudo-polynomial time,
on any input \EG $(V, E, w, \langle V_0, V_1\rangle)$.
\end{Thm}
The algorithm mentioned in Theorem~\ref{Thm:VI} is
the \emph{Value-Iteration} algorithm~\citep{brim2011faster}.
Its rationale relies on the notion of \emph{Small Energy-Progress Measures} (SEPM{s}).

\subsection{Energy-Lattices of Small Energy-Progress Measures}
Small-Energy Progress Measures are bounded, non-negative and integer-valued functions that impose local conditions to ensure global properties on the
arena, in particular, witnessing that Player~0 has a way to enforce conservativity (\ie non-negativity of cycles) in the resulting game's graph.
Recovering standard notation, see e.g.~\cite{brim2011faster}, let us denote $\C_\Gamma=\{n\in\N\mid n\leq (|V|-1) W\}\cup\{\top\}$ and let $\preceq$ be the total order on
$\C_\Gamma$ defined as: $x\preceq y$ \textit{iff} either $y=\top$ or $x,y\in\N$ and $x\leq y$.
In order to cast the minus operation to range over $\C_{\Gamma}$, let us consider an operator $\ominus:\C_\Gamma\times\Z\rightarrow \C_\Gamma$ defined as follows:
\[
a\ominus b \triangleq \left\{
	\begin{array}{ll}
		\max(0, a-b), & \text{ if } a\neq \top \text{ and } a-b\leq (|V|-1)W ; \\
		a\ominus b = \top, & \text{ otherwise.} \\
	\end{array}\right.
\]
Given an \EG $\Gamma$ on vertex set $V = V_0 \cup V_1$, a function $f:V\rightarrow\C_\Gamma$ is
a \emph{Small Energy-Progress Measure} (SEPM) for $\Gamma$ if and only if the following two conditions are met:
\begin{enumerate}
\item if $v\in V_0$, then $f(v)\succeq f(v') \ominus w(v,v')$ for \emph{some} $(v, v')\in E$;
\item if $v\in V_1$, then $f(v)\succeq f(v') \ominus w(v,v')$ for \emph{all} $(v, v')\in E$.
\end{enumerate}

The values of a SEPM, \ie the elements of the image $f(V)$, are called the \emph{energy levels} of $f$.
It is worth to denote by $V_f=\{v\in V\mid f(v)\neq\top\}$ the set of vertices having finite energy.
Given a SEPM $f:V\rightarrow \C_\Gamma$ and a vertex $v\in V_0$,
an arc $(v, v')\in E$ is said to be \emph{compatible} with $f$ whenever $f(v)\succeq f(v')\ominus w(v,v')$;
otherwise $(v, v')$ is said to be \emph{incompatible} with $f$. Moreover,
a positional strategy $\sigma_0\in\Sigma^M_0$ is said to be \emph{compatible} with $f$ whenever: $\forall^{v\in V_0}$
if $\sigma_0(v)=v'$ then $(v,v')\in E$ is compatible with $f$; otherwise, $\sigma_0$ is \emph{incompatible} with $f$.

It is well-known that the family of all the SEPMs of a given $\Gamma$ forms a complete (finite) lattice,
which we denote by $\mathcal{E}_\Gamma$ call it the \emph{Energy-Lattice} of $\Gamma$. Therefore, we shall consider:
\[\mathcal{E}_{\Gamma}\triangleq \big(\{f:V\rightarrow \C_\Gamma \mid f \text{ is SEPM of } \Gamma\}, \sqsubseteq),\]
where for any two SEPMs $f,g$ define $f \sqsubseteq g$ \textit{iff} $\forall{v\in V} f(v)\preceq g(v)$.
Notice that, whenever $f$ and $g$ are SEPM{s}, then so is the \emph{minimum function} defined as:
$\forall^{v\in V} h(v)\triangleq \min\{f(v), g(v)\}$.
This fact allows one to consider the \emph{least} SEPM, namely, the unique SEPM $f^*:V\rightarrow \C_\Gamma$ such that,
for any other SEPM $g:V\rightarrow \C_\Gamma$, the following holds: $\forall^{v\in V} f^*(v)\preceq g(v)$.
Thus, $\mathcal{E}_\Gamma$ is a complete lattice. So, $\mathcal{E}_\Gamma$ enjoys of \emph{Knaster–Tarski Theorem},
which states that the set of fixed-points of a monotone function on a complete lattice is again a complete lattice.

Also concerning SEPMs, we shall rely on the following lemmata.
The first one relates SEPMs to the winning region $\W_0$ of Player~0 in \EG{s}.
\begin{Prop}[\cite{brim2011faster}]\label{prop:least_energy_prog_measure}
Let $\Gamma$ be an \EG. Then the following hold.
\begin{enumerate}
\item If $f$ is any SEPM of the \EG $\Gamma$ and $v\in V_{f}$,
then $v$ is a winning starting position for Player~$0$ in the \EG $\Gamma$.
Stated otherwise, $V_f\subseteq \W_0$;
\item If $f^*$ is the least SEPM of the \EG $\Gamma$,
and $v$ is a winning starting position for Player~$0$ in the \EG $\Gamma$, then $v\in V_{f^*}$.
Thus, $V_{f^*}=\W_0$.
\end{enumerate}
\end{Prop}
The following bound holds on the energy-levels of any SEPM (by definition of $\C_{\Gamma}$).
\begin{Prop}\label{prop:lepm_equals_mincredit}
Let $\Gamma$ be an \EG. Let $f$ be any SEPM of $\Gamma$.

Then, for every $v\in V$ either $f(v)=\top$ or $0\leq f(v)\leq (|V|-1)W$.
\end{Prop}

%% file: Faster.Algo-Sect4-Values-Optimal-Strategies-Reweightings.tex
\section{Optimal Strategies from Reweightings}\label{section:values}

It is now recalled a sufficient condition, for a positional strategy to be optimal, which is expressed in terms of reweighted \EG{s} and their SEPM{s}.
\begin{Thm}[\cite{CR16}]\label{Thm:pos_opt_strategy} Let $\Gamma = (V , E, w, \langle V_0, V_1 \rangle)$ be an \MPG.
For each $u\in V$, consider the reweighted \EG $\Gamma_u \cong \Gamma^{w-\val{\Gamma}{u}}$.
Let $f_{u}:V\rightarrow\C_{\Gamma_u}$ be any SEPM of $\Gamma_u$ such that $u\in V_{f_u}$ (\ie $f_u(u)\neq\top$).
Moreover, we assume: $f_{u_1}=f_{u_2}$ whenever $\val{\Gamma}{u_1}=\val{\Gamma}{u_2}$.

When $u\in V_0$, let $v_{f_u}\in N_{\Gamma}^{\text{out}}(u)$ be any vertex such that $(u, v_{f_u})\in E$ is compatible with $f_u$ in \EG $\Gamma_u$,
and consider the positional strategy $\sigma^*_0\in\Sigma^M_{0}$ defined as follows: $\forall^{u\in V_0}\, \sigma^*_0(u) \triangleq v_{f_u}$.

Then, $\sigma^*_0$ is an optimal positional strategy for Player~$0$ in the \MPG $\Gamma$.
\end{Thm}
\begin{proof}
See the proof of [Theorem~4 in \cite{CR16}].
\end{proof}

\begin{Rem}\label{Rem:pos_opt_strategy} Notice that Theorem~\ref{Thm:pos_opt_strategy} holds, particularly,
when $f_u$ is the least SEPM $f^*_u$ of the reweighted \EG $\Gamma_u$.
This is because $u\in V_{f^*_u}$ always holds for the least SEPM $f^*_u$ of the \EG $\Gamma_u$:
indeed, by Proposition~\ref{prop:relation_MPG_EG} and by definition of $\Gamma_u$,
then $u$ is a winning starting position for Player~0 in the \EG $\Gamma_u$ (for some initial credit);
thus, by Proposition~\ref{prop:least_energy_prog_measure}, it follows that $u\in V_{f^*_u}$.
\end{Rem}

%% file: Faster.Algo-Sect6-EnergyDecomp.tex
\section{An Energy-Lattice Decomposition of $\texttt{opt}_\Gamma\Sigma^M_0$}\label{sect:energy}
Recall the example arena $\Gamma_{\text{ex}}$ shown in \figref{fig:ex1_arena}.
It is easy to see that $\forall^{v\in V} \val{\Gamma_{\text{ex}}}{v}=-1$.
Indeed, $\Gamma_{\text{ex}}$ contains only two cycles, \ie $C_L=[A,B,C,D]$ and $C_R=[F,G]$,
also notice that $w(C_L)/C_L=w(C_R)/C_R=-1$. The least-SEPM $f^*$ of the reweighted \EG $\Gamma_{\text{ex}}^{w+1}$
can be computed by running a Value Iteration~\citep{brim2011faster}.
Taking into account the reweighting $w\leadsto w+1$, as in \figref{fig:ex1_reweighted_leastSEPM}:
$f^*(A)=f^*(E)=f^*(G)=0$, $f^*(B)=f^*(D)=f^*(F)=4$, and $f^*(C)=8$.
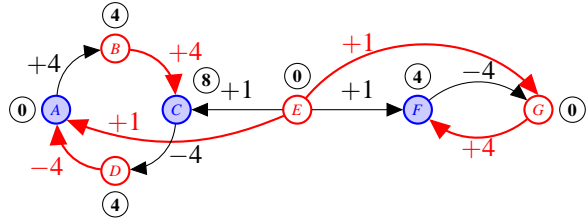
\begin{figure}[!h] \center
\begin{tikzpicture}[scale=.6, arrows={-triangle 45}, node distance=1.5 and 2]
 		\node[node, thick, color=red, label={$\sizedcircled{.15ex}{0}$}] (E) {$E$};
		\node[node, thick, color=blue, fill=blue!20, left=of E,
					label={above right:$\sizedcircled{.15ex}{8}$}] (C) {$C$};
		\node[node, thick, color=red, above=of C, xshift=-8.5ex, yshift=-5ex,
					label={$\sizedcircled{.15ex}{4}$}] (B) {$B$};
		\node[node, thick, color=blue, fill=blue!20, left=of C,
					label={left:$\sizedcircled{.15ex}{0}$}] (A) {$A$};
		\node[node, thick, color=red, below=of C, xshift=-8.5ex, yshift=5ex,
					label={below:$\sizedcircled{.15ex}{4}$}] (D) {$D$};
		\node[node, thick, color=blue, fill=blue!20, right=of E,
					label={$\sizedcircled{.15ex}{4}$}] (F) {$F$};
		\node[node, thick, color=red, right=of F,
					label={right:$\sizedcircled{.15ex}{0}$}] (G) {$G$};
		\draw[] (E) to [bend left=0] node[above] {$+1$} (C);
		\draw[] (E) to [bend left=0] node[above] {$+1$} (F);
		\draw[color=red, thick] (E) to [bend left=22]
					node[above left, xshift=-2ex] {$+1$} (A.south east);
		\draw[color=red, thick] (E) to [bend left=50]
					node[above left, xshift=-3ex, yshift=-1.5ex] {$+1$} (G.north);
		\draw[] (A) to [bend left=40] node[left] {$+4$} (B);
		\draw[color=red, thick] (B) to [bend left=40] node[xshift=2ex, yshift=1ex] {$+4$} (C);
		\draw[] (C) to [bend left=40] node[xshift=2ex, yshift=0ex] {$-4$} (D);
		\draw[color=red, thick] (D) to [bend left=40] node[xshift=-2ex, yshift=-1ex] {$-4$} (A);
		\draw[] (F) to [bend left=40] node[above, yshift=-.75ex] {$-4$} (G);
		\draw[color=red, thick] (G) to [bend left=40] node[below, yshift=.75ex] {$+4$} (F);
\end{tikzpicture}
\caption{The least-SEPM $f^*$ of $\Gamma_{\text{ex}}^{w+1}$ (energy-levels are depicted in circled boldface).
All and only those arcs of Player~0 that are compatible with $f^*$
are $(B,C), (D,A), (E,A), (E,G), (G,F)$ (thick red arcs).}\label{fig:ex1_reweighted_leastSEPM}
\end{figure}

So, $\Gamma_{\text{ex}}$ (\figref{fig:ex1_reweighted_leastSEPM}) implies the following.
\begin{Prop}\label{prop:counter_example}
The converse statement of Theorem~\ref{Thm:pos_opt_strategy} doesn't hold;
there exist infinitely many \MPG{s} $\Gamma$ having at least one $\sigma_0\in\texttt{opt}_{\Gamma}\Sigma^M_0$
which is not compatible with the least-SEPM of $\Gamma$.
\end{Prop}
\begin{proof} Consider the $\Gamma_{\text{ex}}$ of \figref{fig:ex1_reweighted_leastSEPM},
and the least-SEPM $f^*$ of the \EG $\Gamma_{\text{ex}}^{w+1}$. The only vertex at which Player~0 really has a choice is $E$.
Every arc going out of $E$ is optimal in the \MPG $\Gamma_{\text{ex}}$:
whatever arc $(E,X)\in\E$ (for any $X\in \{A,C,F,G\}$) Player~0 chooses at $E$, the resulting payoff
equals $\val{\Gamma_{\text{ex}}}{E}=-1$. Let $f^*$ be the least-SEPM of $f^*$ in $\Gamma_{\text{ex}}^{w+1}$.
Observe, $(E,C)$ and $(E,F)$ are not compatible with $f^*$ in $\Gamma_{\text{ex}}^{w+1}$, only $(E,A)$ and $(E,G)$ are.
For instance, the positional strategy $\sigma_0\in\Sigma^M_0$ defined
as $\sigma_0(E)\triangleq F$, $\sigma_0(B)\triangleq C$, $\sigma_0(D)\triangleq A$, $\sigma_0(G)\triangleq F$ ensures a payoff
$\forall^{v\in V}\val{\Gamma_{\text{ex}}}{v}=-1$, but it is not compatible with the least-SEPM $f^*$ of
$\Gamma_{\text{ex}}^{w+1}$ (because $f^*(E) = 0 < 3 = f^*(F) \ominus w(E,F)$).
It is easy to turn the $\Gamma_{\text{ex}}$ of \figref{fig:ex1_reweighted_leastSEPM}
into a family on infinitely many similar examples.
\end{proof}

We now aim at strengthening the relationship between $\text{opt}_{\Gamma}\Sigma^M_0$
and the Energy-Lattice $\mathcal{E}_\Gamma$. For this, we assume \textit{wlog}
$\exists^{\nu\in\Q} \forall^{v\in V} \val{\Gamma}{v}=\nu$; this follows from Theorem~\ref{thm:ergodic_partition},
which allows one to partition $\Gamma$ into several domains $\Gamma_i \triangleq \Gamma_{|_{C_i}}$
each one satisfying: $\exists^{\nu_i\in\Q}\forall^{v\in C_i} \val{\Gamma_i}{v}=\nu_i$.
By Theorem~\ref{thm:ergodic_partition} we can study $\texttt{opt}_{\Gamma_i}\Sigma^M_0$,
independently \wrt $\texttt{opt}_{\Gamma_j}\Sigma^M_0$ for $j\neq i$.

We say that an \MPG $\Gamma$ is \emph{$\nu$-valued} if and only if $\exists^{\nu\in\Q}\forall^{v\in V} \val{\Gamma}{v}=\nu$.

Given an \MPG $\Gamma$ and $\sigma_0\in\Sigma^M_0(\Gamma)$, recall,
$G(\Gamma,\sigma_0)\triangleq (V, E', w')$ is obtained from $G^\Gamma$
by deleting all and only those arcs that are not part of $\sigma_0$,
\ie \[E' \triangleq \big\{(u,v)\in E \mid u\in V_0 \text{ and }
		v=\sigma_0(u)\big\} \cup \big\{(u,v)\in E \mid u\in V_1\big\}, \]
where each $e\in E'$ is weighted as in $\Gamma$, \ie $w':E'\rightarrow \Z:e\mapsto w_e$.

When $G=(V,E,w)$ is a weighted directed graph, a \emph{feasible-potential (FP)} for $G$
is any map $\pi:V\rightarrow \C_{G}$ s.t.
$\forall^{u\in V}\forall^{v\in N^{\text{out}}(u)} \pi(u)\succeq \pi(v)\ominus w(u,v)$.
The \emph{least}-FP $\pi^*=\pi^*_{G}$ is the (unique) FP s.t., for any other FP $\pi$,
it holds $\forall^{v\in V} \pi^*(v)\preceq \pi(v)$. Given $G$,
the Bellman-Ford algorithm can be used to produce $\pi^*_{G}$ in $O(|V| |E|)$ time.
Let $\pi^*_{G(\Gamma, \sigma_0)}$ be the \emph{least-FP} of $G(\Gamma, \sigma_0)$.
Notice, for every $\sigma_0\in\Sigma_0^M$, the least-FP $\pi^*_{G(\Gamma, \sigma_0)}$
is actually a SEPM for the \EG $\Gamma$; still it can differ from the least-SEPM of $\Gamma$, due to $\sigma_0$.
We consider the following family of strategies.

\begin{Def}[$\Delta^M_0(f, \Gamma)$-Strategies]
Let $\Gamma=\langle V, E, w, (V_0, V_1) \rangle$ and let $f:V\rightarrow\mathcal{C}_{\Gamma}$ be a SEPM for the \EG $\Gamma$.
Let $\Delta_0^M(f, \Gamma)\subseteq \Sigma^M_0(\Gamma)$ be the family of all
and only those positional strategies of Player~0 in $\Gamma$
s.t. $\pi^{*}_{G(\Gamma,\sigma_0)}$ coincides with $f$ pointwisely,
\ie \[\Delta_0^M(f, \Gamma)\triangleq\Big\{\sigma_0\in\Sigma_{0}^M(\Gamma)
	\left|\right. \forall^{ v\in V } \, \pi^*_{G(\Gamma,\sigma_0)}(v)=f(v)\Big\}.\]
\end{Def}

We now aim at exploring further on the relationship between $\mathcal{E}_\Gamma$ and
$\texttt{opt}_{\Gamma}\Sigma^M_0$, via $\Delta_0^M(f, \Gamma)$.
\begin{Def}[The Energy-Lattice of $\texttt{opt}_{\Gamma}\Sigma^M_0$]
Let $\Gamma$ be a $\nu$-valued \MPG. Let $\mathcal{X} \subseteq \mathcal{E}_{\Gamma^{w-\nu}}$ be a sublattice
of SEPM{s} of the reweighted \EG $\Gamma^{w-\nu}$.

We say that $\mathcal{X}$ is an ``\emph{Energy-Lattice} of $\texttt{opt}_{\Gamma}\Sigma^M_0$"
\textit{iff} $\forall^{f\in\mathcal{X}} \Delta_0^M(f, \Gamma^{w-\nu})\neq\emptyset$ and the following disjoint-set
decomposition holds: \[\displaystyle \texttt{opt}_{\Gamma}\Sigma^M_0 =
				\bigsqcup_{f\in\mathcal{X}} \Delta_0^M(f, \Gamma^{w-\nu}).\]
\end{Def}

\begin{Lem}\label{lem:pre_main_thm}
Let $\Gamma$ be a $\nu$-valued \MPG, and let $\sigma^*_0\in\texttt{opt}_{\Gamma}\Sigma^M_0$. Then,
$G(\Gamma^{w-\nu},\sigma^*_0)$ is \emph{conservative} (\ie it contains \emph{no} negative cycle).
\end{Lem}
\begin{proof}
Let $C \triangleq (v_1\ldots, v_{k}, v_1)$ by any cycle in $G(\Gamma^{w-\nu},\sigma^*_0)$.
Since we have $\sigma^*_0\in\texttt{opt}_\Gamma\Sigma^M_0$ and $\forall^{v\in V} \val{\Gamma}{v}=\nu$, thus
$w(C)/k=\frac{1}{k}\sum_{i=1}^{k} w(v_i, v_{i+1})
	\geq \nu$ (for $v_{k+1}\triangleq v_1$) by Proposition~\ref{prop:reachable_cycle},
so that, assuming $w'\triangleq w-\nu$, then: $w'(C)/k=\frac{1}{k}\sum_{i=1}^{k}
	\big(w(v_i, v_{i+1})-\nu\big) = w(C)/k - \nu \geq \nu - \nu = 0$.
\end{proof}

Some aspects of the following Proposition~\ref{prop:energy_existence} rely heavily on Theorem~\ref{Thm:pos_opt_strategy}:
the compatibility relation comes again into play.
Moreover, we observe that Proposition~\ref{prop:energy_existence} is equivalent to the following fact,
which provides a sufficient condition for a positional strategy to be optimal.
Consider a $\nu$-valued \MPG $\Gamma$, for some $\nu\in \Q$,
and let $\sigma^*_0\in\texttt{opt}_{\Gamma}\Sigma^M_0$.
Let $\hat{\sigma}_0\in\Sigma^M_0(\Gamma)$ be any (not necessarily optimal) positional strategy for Player~0 in the \MPG $\Gamma$.
Suppose the following holds: \[\forall^{v\in V} \pi^*_{G(\Gamma^{w-\nu}, \hat{\sigma}_0)}(v)=\pi^*_{G(\Gamma^{w-\nu},\sigma^*_0)}(v).\]
Then, by Proposition~\ref{prop:energy_existence}, $\hat{\sigma_0}$
	is an optimal positional strategy for Player~0 in the \MPG $\Gamma$.

We are thus relying on the same \emph{compatibility} relation between $\Sigma^M_0$ and SEPM{s} in reweighted \EG{s}
which was at the \emph{base} of Theorem~\ref{Thm:pos_opt_strategy}, aiming at extending Theorem~\ref{Thm:pos_opt_strategy}
so to describe the whole $\texttt{opt}_{\Gamma}\Sigma^M_0$ (and not just the join/top component of it).

\begin{Prop}\label{prop:energy_existence}
Let the \MPG $\Gamma$ be $\nu$-valued, for some $\nu\in \Q$.

There is at least one Energy-Lattice of $\texttt{opt}_{\Gamma}\Sigma^M_0$:
\[\mathcal{X}^*_{\Gamma}\triangleq \{ \pi^*_{G(\Gamma^{w-\nu},\sigma_0)}
	\mid {\sigma_0\in \texttt{opt}_{\Gamma}\Sigma_0^M}\}.\]
\end{Prop}
\begin{proof}
The only non-trivial point to check being:
$ \bigsqcup_{f\in\mathcal{X}^*_{\Gamma}} \Delta^M_0(f, \Gamma^{w-\nu}) \subseteq \texttt{opt}_\Gamma\Sigma^M_0 $.

For this, we shall rely on Theorem~\ref{Thm:pos_opt_strategy}.
Let $\hat{f}\in \mathcal{X}^*_{\Gamma}$ and $\hat{\sigma}_0\in \Delta^M_0(\hat{f}, \Gamma^{w-\nu})$ be fixed (arbitrarily).
Since $\hat{f}\in\mathcal{X}^*_{\Gamma}$, then
$\hat{f}=\pi^*_{G(\Gamma^{w-\nu}, \sigma^*_0)}$ for some $\sigma^*_0\in\texttt{opt}_\Gamma\Sigma^M_0$.
Therefore, the following holds:
	\[ \pi^*_{G(\Gamma^{w-\nu}, \hat{\sigma}_0)} = \hat{f} = \pi^*_{G(\Gamma^{w-\nu}, \sigma^*_0)}. \]
Clearly, $\hat{\sigma}_0$ is compatible
with $\hat{f}$ in the \EG $\Gamma^{w-\nu}$, because $\hat{f}=\pi^*_{G(\Gamma^{w-\nu}, \hat{\sigma}_0)}$.
By Lemma~\ref{lem:pre_main_thm}, since $\sigma^*_0$ is optimal,
then $G(\Gamma^{w-\nu}, \sigma^*_0)$ is conservative. Therefore:
\[V_{\hat{f}}=V_{\pi^*_{G(\Gamma^{w-\nu}, \sigma^*_0)}}=V.\]
Notice, $\hat{\sigma}_0$ satisfies exactly the hypotheses required by Theorem~\ref{Thm:pos_opt_strategy}.
Therefore, $\hat{\sigma}_0\in\texttt{opt}_\Gamma\Sigma^M_0$.
This proves (*).
This also shows $\texttt{opt}_{\Gamma}\Sigma^M_0 =
	\bigsqcup_{f\in\mathcal{X}^*_{\Gamma}} \Delta^M_0(f, \Gamma^{w-\nu})$, and concludes the proof.
\end{proof}

\begin{Prop}\label{prop:energy_uniqueness}
Let the \MPG $\Gamma$ be $\nu$-valued, for some $\nu\in \Q$. Let ${\mathcal{X}^*_{\Gamma}}_1$ and
${\mathcal{X}^*_{\Gamma}}_2$ be two Energy-Lattices for $\texttt{opt}_{\Gamma}\Sigma^M_0$.
Then, ${\mathcal{X}^*_{\Gamma}}_1={\mathcal{X}^*_{\Gamma}}_2$.
\end{Prop}
\begin{proof}
By symmetry, it is sufficient to prove that ${\mathcal{X}^*_{\Gamma}}_1\subseteq {\mathcal{X}^*_{\Gamma}}_2$. Let $f_1\in{\mathcal{X}^*_{\Gamma}}_1$ be fixed (arbitrarily).
Then, $f_1=\pi^*_{G(\Gamma^{w-\nu}, \hat{\sigma}_0)}$
for some $\hat{\sigma}_0\in\texttt{opt}_{\Gamma}\Sigma^M_0$.
Since $\hat{\sigma}_0\in\texttt{opt}_{\Gamma}\Sigma^M_0$ and since ${\mathcal{X}^*_{\Gamma}}_2$ is an Energy-Lattices,
there exists $f_2\in {\mathcal{X}^*_{\Gamma}}_2$ s.t. $\hat{\sigma}_0\in\Delta^M_0(f_2, \Gamma^{w-\nu})$, which implies
$\pi^*_{G(\Gamma^{w-\nu},\hat{\sigma}_0)}=f_2$. Thus, $f_1= \pi^*_{G(\Gamma^{w-\nu},\hat{\sigma}_0)} = f_2$.
This implies $f_1\in {\mathcal{X}^*_{\Gamma}}_2$.
\end{proof}

The next theorem summarizes the main point of this section.
\begin{Thm}\label{thm:main_energystructure}
Let $\Gamma$ be a $\nu$-valued \MPG, for some $\nu\in \Q$.
Then, $\mathcal{X}^*_{\Gamma}\triangleq \{ \pi^*_{G(\Gamma^{w-\nu},\sigma_0)}
				\mid {\sigma_0\in \texttt{opt}_{\Gamma}\Sigma_0^M}\}$
	is the unique Energy-Lattice of $\texttt{opt}_\Gamma\Sigma^M_0$.
\end{Thm}
\begin{proof}
By Proposition~\ref{prop:energy_existence} and Proposition~\ref{prop:energy_uniqueness}.
\end{proof}

\begin{Exp}\label{exp1}
Consider the \MPG $\Gamma_{\text{ex}}$, as defined in \figref{fig:ex1_arena}. Then,
$\mathcal{X}^*_{\Gamma_{\texttt{ex}}}=\{f^*, f_1, f_2\}$,
where $f^*$ is the least-SEPM of the reweighted \EG $\Gamma^{w+1}_{\texttt{ex}}$,
and where the following holds:
$f_1(A)=f_2(A)=f^*(A)=0$; $f_1(B)=f_2(B)=f^*(B)=4$; $f_1(C)=f_2(C)=f^*(C)=8$; $f_1(D)=f_2(D)=f^*(D)=4$;
$f_1(F)=f_2(F)=f^*(F)=4$; $f_1(G)=f_2(G)=f^*(G)=0$; finally, $f^*(E)=0$, $f_1(E)=3$, $f_2(E)=7$.
An illustration of $f_1$ is offered in \figref{fig:ex1_reweighted_f1SEPM} (energy-levels are depicted in circled boldface).
whereas $f_2$ is depicted in \figref{fig:ex1_reweighted_f2SEPM}.
Notice that $f^*(v)\leq f_1(v)\leq f_2(v)$ for every $v\in V$,
and this ordering relation is illustrated in \figref{fig:ex1_ordering}.
\end{Exp}

\begin{figure}[!h]
\center
\subfloat[The extremal-SEPM $f_1$ of $\Gamma_{\text{ex}}^{w+1}$]{\label{fig:ex1_reweighted_f1SEPM}
\begin{tikzpicture}[scale=.57, arrows={-triangle 45}, node distance=1.5 and 1.5]
 		\node[node, thick, color=red, label={$\sizedcircled{.15ex}{3}$}] (E) {$E$};
		\node[node, thick, color=blue, fill=blue!20, left=of E, label={above right:$\sizedcircled{.15ex}{8}$}] (C) {$C$};
		\node[node, thick, color=red, above=of C, xshift=-8.5ex, yshift=-5ex, label={$\sizedcircled{.15ex}{4}$}] (B) {$B$};
		\node[node, thick, color=blue, fill=blue!20, left=of C, xshift=-4ex, label={left:$\sizedcircled{.15ex}{0}$}] (A) {$A$};
		\node[node, thick, color=red, below=of C, xshift=-8.5ex, yshift=5ex, label={below:$\sizedcircled{.15ex}{4}$}] (D) {$D$};
		\node[node, thick, color=blue, fill=blue!20, right=of E, label={$\sizedcircled{.15ex}{4}$}] (F) {$F$};
		\node[node, thick, color=red, right=of F, label={right:$\sizedcircled{.15ex}{0}$}] (G) {$G$};
		\draw[] (E) to [bend left=0] node[above, xshift=1ex] {$+1$} (C);
		\draw[color=red, thick] (E) to [bend left=0] node[above] {$+1$} (F);
		\draw[dotted] (E) to [bend left=22] node[above left, xshift=-2ex] {$+1$} (A.south east);
		\draw[dotted] (E) to [bend left=50] node[above left, xshift=-4ex] {$+1$} (G.north);
		\draw[] (A) to [bend left=40] node[left] {$+4$} (B);
		\draw[color=red, thick] (B) to [bend left=40] node[xshift=2ex, yshift=1ex] {$+4$} (C);
		\draw[] (C) to [bend left=40] node[xshift=2ex, yshift=0ex] {$-4$} (D);
		\draw[color=red, thick] (D) to [bend left=40] node[xshift=-2ex, yshift=-1ex] {$-4$} (A);
		\draw[] (F) to [bend left=40] node[above, xshift=-1ex, yshift=2ex] {$-4$} (G);
		\draw[color=red, thick] (G) to [bend left=40] node[below] {$+4$} (F);
\end{tikzpicture}
}
\subfloat[The extremal-SEPM $f_2$ of $\Gamma_{\text{ex}}^{w+1}$.]{\label{fig:ex1_reweighted_f2SEPM}
\begin{tikzpicture}[scale=.57, arrows={-triangle 45}, node distance=1.5 and 1.5]
 		\node[node, thick, color=red, label={$\sizedcircled{.15ex}{7}$}] (E) {$E$};
		\node[node, thick, color=blue, fill=blue!20, left=of E, label={above right:$\sizedcircled{.15ex}{8}$}] (C) {$C$};
		\node[node, thick, color=red, above=of C, xshift=-8.5ex, yshift=-5ex, label={$\sizedcircled{.15ex}{4}$}] (B) {$B$};
		\node[node, thick, color=blue, fill=blue!20, left=of C, xshift=-4ex, label={left:$\sizedcircled{.15ex}{0}$}] (A) {$A$};
		\node[node, thick, color=red, below=of C, xshift=-8.5ex, yshift=5ex, label={below:$\sizedcircled{.15ex}{4}$}] (D) {$D$};
		\node[node, thick, color=blue, fill=blue!20, right=of E, label={$\sizedcircled{.15ex}{4}$}] (F) {$F$};
		\node[node, thick, color=red, right=of F, label={right:$\sizedcircled{.15ex}{0}$}] (G) {$G$};
		\draw[color=red, thick] (E) to [bend left=0] node[above, xshift=1ex] {$+1$} (C);
		\draw[dotted] (E) to [bend left=0] node[above] {$+1$} (F);
		\draw[dotted] (E) to [bend left=22] node[above left, xshift=-2ex] {$+1$} (A.south east);
		\draw[dotted] (E) to [bend left=50] node[above left, xshift=-4ex] {$+1$} (G.north);
		\draw[] (A) to [bend left=40] node[left] {$+4$} (B);
		\draw[color=red, thick] (B) to [bend left=40] node[xshift=2ex, yshift=1ex] {$+4$} (C);
		\draw[] (C) to [bend left=40] node[xshift=2ex, yshift=0ex] {$-4$} (D);
		\draw[color=red, thick] (D) to [bend left=40] node[xshift=-2ex, yshift=-1ex] {$-4$} (A);
		\draw[] (F) to [bend left=40] node[above, xshift=-1ex, yshift=2ex] {$-4$} (G);
		\draw[color=red, thick] (G) to [bend left=40] node[below] {$+4$} (F);
\end{tikzpicture}
}
\end{figure}

\begin{Def}
Each element $f\in\mathcal{X}^*_{\Gamma}$ is called \emph{extremal}-SEPM.
\end{Def}

The next lemma is the converse of Lemma~\ref{lem:pre_main_thm}.
\begin{Lem}\label{lem:conservative_implies_optimal}
Let the \MPG $\Gamma$ be $\nu$-valued, for some $\nu\in \Q$.
Consider any $\sigma_0\in\Sigma_0^M(\Gamma)$, and assume that $G(\Gamma^{w-\nu}, \sigma_0)$ is conservative.
Then, $\sigma_0\in\texttt{opt}_\Gamma\Sigma_0^M$.
\end{Lem}
\begin{proof}
Let $C=(v_1, \ldots, v_{\ell}v_1)$ any cycle in $G(\Gamma, \sigma_0)$.
Then, the following holds (if $v_{\ell+1}=v_1$):
$\frac{w(C)}{\ell}  = \frac{1}{\ell}\sum_{i=1}^{\ell} w(v_i, v_{i+1})
		   = \nu + \frac{1}{\ell}\sum_{i=1}^{\ell} \Big( w(v_i, v_{i+1})-\nu \Big) \geq \nu$,
where $\frac{1}{\ell}\sum_{i=1}^{\ell} \big( w(v_i, v_{i+1})-\nu \big) \geq 0$
holds because $G(\Gamma^{w-\nu}, \sigma_0)$ is conservative.
By Proposition~\ref{prop:reachable_cycle}, since $w(C)/\ell \geq \nu$ for
	every cycle $C$ in $G^{\Gamma}_{\sigma_0}$, then $\sigma_0\in\texttt{opt}_\Gamma\Sigma_0^M$.
\end{proof}

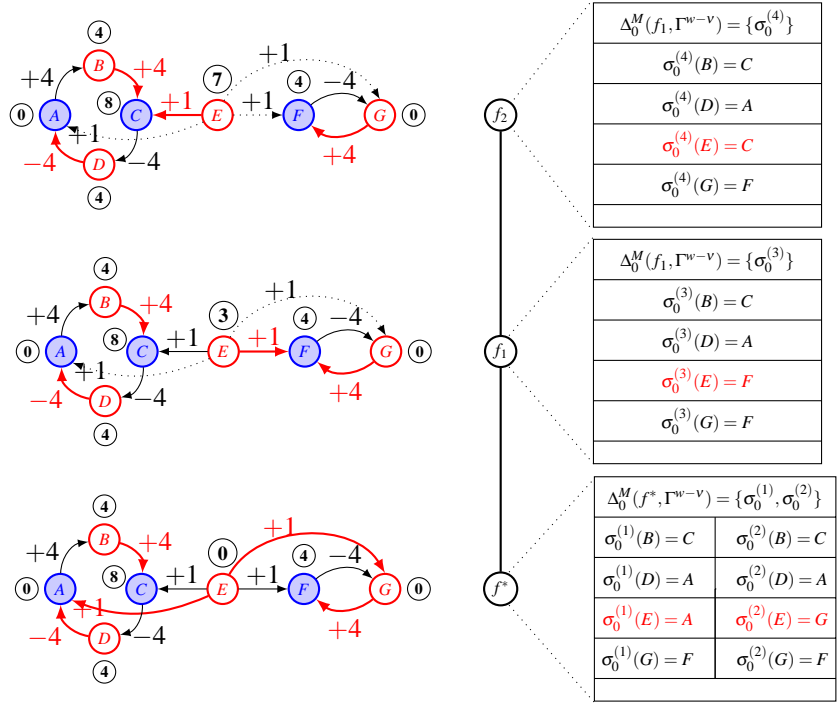
\begin{figure}[!h]
\center
\begin{tikzpicture}[scale=.65, node distance=1 and 1]
 		\node[node, thick, color=red, label={$\sizedcircled{.75}{7}$}] (E2) {$E$};
\node[node, thick, color=blue, fill=blue!20,  left=of E2, label={left, xshift=.5ex, yshift=1ex:$\sizedcircled{.6}{8}$}] (C2) {$C$};
		\node[node, thick, color=red, above=of C2, xshift=-4.75ex, yshift=-4ex, label={$\sizedcircled{.6}{4}$}] (B2) {$B$};
		\node[node, thick, color=blue, fill=blue!20, left=of C2, label={left:$\sizedcircled{.6}{0}$}] (A2) {$A$};
		\node[node, thick, color=red, below=of C2, xshift=-4.75ex, yshift=4ex, label={below:$\sizedcircled{.6}{4}$}] (D2) {$D$};
		\node[node, thick, color=blue, fill=blue!20, right=of E2, label={$\sizedcircled{.6}{4}$}] (F2) {$F$};
		\node[node, thick, color=red, right=of F2, label={right:$\sizedcircled{.6}{0}$}] (G2) {$G$};
		\draw[->, color=red, thick] (E2) to [bend left=0] node[above, yshift=-.5ex] {$+1$} (C2);
		\draw[->, dotted] (E2) to [bend left=0] node[above, yshift=-.5ex] {$+1$} (F2);
		\draw[->, dotted] (E2) to [bend left=22] node[above left, xshift=-2ex, yshift=-1.25ex] {$+1$} (A2.south east);
		\draw[->, dotted] (E2) to [bend left=60] node[above left, xshift=0ex, yshift=-1ex] {$+1$} (G2.north);
		\draw[->] (A2) to [bend left=40] node[left] {$+4$} (B2);
		\draw[->,color=red, thick] (B2) to [bend left=40] node[xshift=1.5ex, yshift=1ex] {$+4$} (C2);
		\draw[->] (C2) to [bend left=40] node[xshift=1ex, yshift=-1ex] {$-4$} (D2);
		\draw[->,color=red, thick] (D2) to [bend left=40] node[xshift=-2ex, yshift=-1ex] {$-4$} (A2);
		\draw[->] (F2) to [bend left=40] node[above, yshift=-.5ex] {$-4$} (G2);
		\draw[->,color=red, thick] (G2) to [bend left=40, yshift=1.5ex] node[below] {$+4$} (F2);
		\node[node, thick, right=of G2, xshift=5ex] (fA) {$f_2$};
		\node[node, thick, below=of fA, yshift=-20ex] (fB) {$f_1$};
 		\node[node, thick, left=of fB, xshift=-25ex, color=red, label={$\sizedcircled{.75}{3}$}] (E1) {$E$};
\node[node, thick, color=blue, fill=blue!20, left=of E1, label={left, xshift=.5ex, yshift=1ex:$\sizedcircled{.6}{8}$}] (C1) {$C$};
		\node[node, thick, color=red, above=of C1, xshift=-4.75ex, yshift=-4ex, label={$\sizedcircled{.6}{4}$}] (B1) {$B$};
		\node[node, thick, color=blue, fill=blue!20, left=of C1, label={left:$\sizedcircled{.6}{0}$}] (A1) {$A$};
		\node[node, thick, color=red, below=of C1, xshift=-4.75ex, yshift=4ex, label={below:$\sizedcircled{.6}{4}$}] (D1) {$D$};
		\node[node, thick, color=blue, fill=blue!20, right=of E1, label={$\sizedcircled{.6}{4}$}] (F1) {$F$};
		\node[node, thick, color=red, right=of F1, label={right:$\sizedcircled{.6}{0}$}] (G1) {$G$};
		\draw[->] (E1) to [bend left=0] node[above, yshift=-.5ex] {$+1$} (C1);
		\draw[->, color=red, thick] (E1) to [bend left=0] node[above, yshift=-.5ex] {$+1$} (F1);
		\draw[->, dotted] (E1) to [bend left=22] node[above left, xshift=-2ex, yshift=-1.25ex] {$+1$} (A1.south east);
		\draw[->, dotted] (E1) to [bend left=60] node[above left, xshift=0ex, yshift=-1ex] {$+1$} (G1.north);
		\draw[->] (A1) to [bend left=40] node[left] {$+4$} (B1);
		\draw[->,color=red, thick] (B1) to [bend left=40] node[xshift=1.5ex, yshift=1ex] {$+4$} (C1);
		\draw[->] (C1) to [bend left=40] node[xshift=1ex, yshift=-1ex] {$-4$} (D1);
		\draw[->,color=red, thick] (D1) to [bend left=40] node[xshift=-2ex, yshift=-1ex] {$-4$} (A1);
		\draw[->] (F1) to [bend left=40] node[above, yshift=-.5ex] {$-4$} (G1);
		\draw[->,color=red, thick] (G1) to [bend left=40, yshift=1.5ex] node[below] {$+4$} (F1);
		\node[node, thick, below=of fB, yshift=-20ex] (fC) {$f^*$};
 		\node[node, thick, left=of fC, xshift=-25ex, color=red, label={$\sizedcircled{.75}{0}$}] (E0) {$E$};
\node[node, thick, color=blue, fill=blue!20, left=of E0, label={left, xshift=.5ex, yshift=1ex:$\sizedcircled{.6}{8}$}] (C0) {$C$};
		\node[node, thick, color=red, above=of C0, xshift=-4.75ex, yshift=-4ex, label={$\sizedcircled{.6}{4}$}] (B0) {$B$};
		\node[node, thick, color=blue, fill=blue!20, left=of C0, label={left:$\sizedcircled{.6}{0}$}] (A0) {$A$};
		\node[node, thick, color=red, below=of C0, xshift=-4.75ex, yshift=4ex, label={below:$\sizedcircled{.6}{4}$}] (D0) {$D$};
		\node[node, thick, color=blue, fill=blue!20, right=of E0, label={$\sizedcircled{.6}{4}$}] (F0) {$F$};
		\node[node, thick, color=red, right=of F0, label={right:$\sizedcircled{.6}{0}$}] (G0) {$G$};
	\draw[->] (E0) to [bend left=0] node[above, yshift=-.5ex] {$+1$} (C0);
	\draw[->] (E0) to [bend left=0] node[above, yshift=-.5ex] {$+1$} (F0);
	\draw[->, color=red, thick] (E0) to [bend left=22] node[above left, xshift=-2ex, yshift=-1.25ex] {$+1$} (A0.south east);
	\draw[->, color=red, thick] (E0) to [bend left=60] node[above left, xshift=0ex, yshift=-1ex] {$+1$} (G0.north);
	\draw[->] (A0) to [bend left=40] node[left] {$+4$} (B0);
	\draw[->,color=red, thick] (B0) to [bend left=40] node[xshift=1.5ex, yshift=1ex] {$+4$} (C0);
	\draw[->] (C0) to [bend left=40] node[xshift=1ex, yshift=-1ex] {$-4$} (D0);
	\draw[->,color=red, thick] (D0) to [bend left=40] node[xshift=-2ex, yshift=-1ex] {$-4$} (A0);
	\draw[->] (F0) to [bend left=40] node[above, yshift=-.5ex] {$-4$} (G0);
	\draw[->,color=red, thick] (G0) to [bend left=40, yshift=1.5ex] node[below] {$+4$} (F0);
		\draw[thick] (fA) to [] node[] {} (fB);
		\draw[thick] (fB) to [] node[] {} (fC);

		\node[scale=0.75, right=of fC, yshift=0ex, rectangle split, rectangle split parts=6, draw,
			minimum width=4cm, font=\small, rectangle split part align={center}] (t1)
  		{
		  {$\Delta^M_0(f^*, \Gamma^{w-\nu})=\{\sigma^{(1)}_0, \sigma^{(2)}_0\}$}
		     \nodepart{two}
       		             $\sigma^{(1)}_0(B)=C$ \hspace*{4ex} $\sigma^{(2)}_0(B)=C$
		     \nodepart{three}
       	  	             $\sigma^{(1)}_0(D)=A$ \hspace*{4ex} $\sigma^{(2)}_0(D)=A$
	     	     \nodepart{four}
                             \textcolor{red}{$\sigma^{(1)}_0(E)=A$} \hspace*{4ex} \textcolor{red}{$\sigma^{(2)}_0(E)=G$}
	 	     \nodepart{five}
                  	     $\sigma^{(1)}_0(G)=F$ \hspace*{4ex} $\sigma^{(2)}_0(G)=F$
 	     	     \nodepart{six}
		};
	 	 \draw (t1.text split) -- (t1.two split);
		 \draw (t1.two split) -- (t1.three split);
		 \draw (t1.three split) -- (t1.four split);
		 \draw (t1.four split) -- (t1.five split);
		 \draw (t1.five split) -- (t1.six split);
		 \draw (t1.six split) -- (t1.seven split);
 		 \draw (t1.seven split) -- (t1.eight split);

	\node[right=of fB, scale=0.75, yshift=0ex, rectangle split, rectangle split parts=6, draw,
			minimum width=4cm, font=\small, rectangle split part align={center}] (t2)
  		  {
		{$\Delta^M_0(f_1, \Gamma^{w-\nu})=\{\sigma^{(3)}_0\}$}
	     \nodepart{two}
     	              $\sigma^{(3)}_0(B)=C$
	     \nodepart{three}
                      $\sigma^{(3)}_0(D)=A$
	     \nodepart{four}
                      \textcolor{red}{$\sigma^{(3)}_0(E)=F$}
	     \nodepart{five}
                      $\sigma^{(3)}_0(G)=F$
 	     \nodepart{six}
		};

		\node[right=of fA, yshift=0ex, scale=0.75, rectangle split, rectangle split parts=6, draw,
			minimum width=4cm, font=\small, rectangle split part align={center}] (t2)
 		  {
		{$\Delta^M_0(f_1, \Gamma^{w-\nu})=\{\sigma^{(4)}_0\}$}
	     \nodepart{two}
     	              $\sigma^{(4)}_0(B)=C$
	     \nodepart{three}
                      $\sigma^{(4)}_0(D)=A$
	     \nodepart{four}
                      \textcolor{red}{$\sigma^{(4)}_0(E)=C$}
	     \nodepart{five}
                      $\sigma^{(4)}_0(G)=F$
 	     \nodepart{six}
		};
    \draw[dotted] (6,0.32) -- (17:7.9cm);
    \draw[dotted] (6,-0.32) -- (-17:7.9cm);
    \draw[dotted] (5.9,-4.5) -- (7.64, -2.45);
    \draw[dotted] (5.9,-5.14) -- (7.64, -7.15);
    \draw[dotted] (5.9,-9.325) -- (7.64, -7.4);
    \draw[dotted] (5.9,-10) -- (7.64, -12);
\end{tikzpicture}
\caption{The decomposition of $\texttt{opt}_\Gamma\Sigma^M_0$ (right), for the \MPG $\Gamma_{\text{ex}}$,
which corresponds to the Energy-Lattice $\mathcal{X}^*_{\Gamma_{\text{ex}}}=\{f^*, f_1, f_2\}$ (center)
(as in Example~\ref{exp1}). Here, $f^*\leq f_1\leq f_2$. This brings a lattice
$\mathcal{D}^*_{\Gamma_{\text{ex}}}$ of 3 basic subgames of $\Gamma_{\text{ex}}$ (left). }\label{fig:ex1_ordering}
\end{figure}

The following proposition asserts some properties of the extremal-SEPM{s}.
\begin{Prop}\label{prop:extremal}
Let the \MPG $\Gamma$ be $\nu$-valued, for some $\nu\in \Q$.
Let $\mathcal{X}^*_{\Gamma}$ be the Energy-Lattice of $\texttt{opt}_\Gamma\Sigma^M_0$.
Moreover, let $f:V\rightarrow \C_{\Gamma}$ be a SEPM for the reweighted \EG $\Gamma^{w-\nu}$.
Then, the following three properties are equivalent:
\begin{enumerate}
  \item $f\in \mathcal{X}^*_{\Gamma}$;
  \item There exists $\sigma_0\in\texttt{opt}_\Gamma\Sigma^M_0$ s.t.
    $\pi^*_{G(\Gamma^{w-\nu}, \sigma_0)}(v)=f(v)$ for every $v\in V$.
  \item $V_f=\mathcal{W}_0(\Gamma^{w-\nu})=V$ and $\Delta^M_0(f, \Gamma^{w-\nu})\neq\emptyset$;
\end{enumerate}
\end{Prop}

\begin{proof}[Proof of ($1\iff 2$)] Indeed, $\mathcal{X}^*_{\Gamma} =
	\{ \pi^*_{G(\Gamma^{w-\nu}, \sigma_0)}\mid {\sigma_0\in \texttt{opt}_{\Gamma}\Sigma_0^M}\} $.
\end{proof}
\begin{proof}[Proof of ($1\Rightarrow 3$)] Assume $f\in \mathcal{X}^*_{\Gamma}$. Since ($1\iff 2$), there exist $\sigma_0\in\texttt{opt}_\Gamma\Sigma_0^M$
s.t. $\pi^*_{G(\Gamma^{w-\nu}, \sigma_0)}=f$. Thus, $\sigma_0\in\Delta^M_0(f, \Gamma^{w-\nu})$, so that $\Delta^M_0(f, \Gamma^{w-\nu})\neq\emptyset$.
We claim $V_f=\mathcal{W}_0(\Gamma^{w-\nu})=V$. Since $\forall (v\in V)\, \val{\Gamma}{v}=\nu$, then $\mathcal{W}_0(\Gamma^{w-\nu})=V$ by Proposition~\ref{prop:relation_MPG_EG}.
Next, $G(\Gamma^{w-\nu}, \sigma_0)$ is conservative by Lemma~\ref{lem:pre_main_thm}.
Since $G(\Gamma^{w-\nu}, \sigma_0)$ is conservative and $f=\pi^*_{G(\Gamma^{w-\nu}, \sigma_0)}$, then $V_f=V$. Therefore, $V_f=\mathcal{W}_0(\Gamma^{w-\nu})=V$.
\end{proof}
\begin{proof}[Proof of ($1\Leftarrow 3$)]
Since $\Delta^M_0(f, \Gamma^{w-\nu})\neq\emptyset$, pick some $\sigma_0\in\Delta^M_0(f, \Gamma^{w-\nu})$;
so, $f=\pi^*_{G(\Gamma^{w-\nu}, \sigma_0)}$. Since $V_f=V$ and $f=\pi^*_{G(\Gamma^{w-\nu}, \sigma_0)}$, then $G(\Gamma^{w-\nu}, \sigma_0)$ is conservative.
Since $G(\Gamma^{w-\nu}, \sigma_0)$ is conservative, then $\sigma_0\in\texttt{opt}_\Gamma\Sigma_0^M$ by Lemma~\ref{lem:conservative_implies_optimal}.
Since $f=\pi^*_{G^*}$ and $\sigma_0\in\texttt{opt}_\Gamma\Sigma_0^M$, then $f\in\mathcal{X}^*_{\Gamma}$ because $2\Rightarrow 1$.
\end{proof}

%% file: Faster.Algo-Sect7-Listing.tex
\section{A Recursive Enumeration of $\mathcal{X}^*_{\Gamma}$
	and $\texttt{opt}_\Gamma\big(\Sigma_{0}^M\big)$}\label{sect:recursive_enumeration}

An enumeration algorithm for a set $S$ provides an exhaustive listing of all the elements of $S$ (without repetitions).
As mentioned in Section~\ref{sect:energy}, by Theorem~\ref{thm:ergodic_partition},
no loss of generality occurs if we assume $\Gamma$ to be $\nu$-valued for some $\nu\in\Q$.
One run of the algorithm given in~\cite{CominR16a} allows one to partition an \MPG $\Gamma$,
into several domains $\Gamma_{i}$ each one being $\nu_{i}$-valued for $\nu_{i}\in S_\Gamma$;
in $O(|V|^2|E|W)$ time and linear space. Still, by Proposition~\ref{prop:counter_example},
Theorem~\ref{Thm:pos_opt_strategy} is not sufficient for enumerating the whole $\texttt{opt}_\Gamma(\Sigma^M_0)$;
it is enough only for $\Delta^M_0(f^*_\nu,\Gamma^{w-\nu})$ where $f^*_\nu$ is the least-SEPM of $\Gamma^{w-\nu}$,
which is just the join/top component of $\texttt{opt}_\Gamma(\Sigma^M_0)$.
However, thanks to Theorem~\ref{thm:main_energystructure},
	we now have a refined description of $\texttt{opt}_\Gamma \Sigma^M_0$ in terms $\mathcal{X}^*_{\Gamma}$.

We offer a recursive enumeration of all the extremal-SEPM{s}, \ie $\mathcal{X}^*_{\Gamma}$,
and for computing the corresponding partitioning of $\texttt{opt}_\Gamma\big(\Sigma_{0}^M\big)$.
In order to avoid duplicate elements in the enumeration,
	the algorithm needs to store a lattice $\mathcal{B}^*_{\Gamma}$ of subgames of $\Gamma$,	which is related to $\mathcal{X}^*_\Gamma$.
We assume to have a data-structure $T_{\Gamma}$ supporting the following operations,
given a subarena $\Gamma'$ of $\Gamma$: $\texttt{insert}(\Gamma', T_{\Gamma})$ stores $\Gamma'$ into $T_{\Gamma}$;
$\texttt{contains}(\Gamma', T_{\Gamma})$ returns $\texttt{T}$ if and only if $\Gamma'$ is in $T_{\Gamma}$, and $\texttt{F}$ otherwise.
A simple implementation of $T_{\Gamma}$ goes by indexing $N^{\text{out}}_{\Gamma'}(v)$ for each $v\in V$ (\eg with a trie data-structure).
This can run in $O(|E|\log|V|)$ time, consuming $O(|E|)$ space per stored item. Similarly,
	one can index SEPM{s} in $O(|V|\log(|V|W))$ time and $O(|V|)$ space per stored item.

The listing procedure is named $\texttt{enum()}$, it takes a $\nu$-valued \MPG $\Gamma$ and goes as follows.
\begin{enumerate}
\item Compute the least-SEPM $f^*$ of $\Gamma$, and \texttt{print} $\Gamma$ to output.
	Theorem~\ref{Thm:pos_opt_strategy} can be employed at this
	stage for enumerating $\Delta^M_0(f^*, \Gamma^{w-\nu})$:
	indeed, these are all and only those positional strategies
	lying in the \emph{Cartesian} product of all the arcs $(u,v)\in E$
	that are \emph{compatible} with $f^*$ in $\Gamma^{w-\nu}$ (because $f^*$ is the least-SEPM of $\Gamma$).
\item Let $\texttt{St}\leftarrow\emptyset$ be an empty stack of vertices.
\item For each $\hat{u}\in V_0$, do the following:
	\begin{itemize}
		\item Compute $E_{\hat{u}}\leftarrow \{(\hat{u},v)\in E\mid f^*(\hat{u})
			\prec f^*(v)\ominus(w(\hat{u},v)-\nu)\}$;
			\item If $E_{\hat{u}}\neq \emptyset$, then:
		\begin{itemize}
		\item Let $E'\leftarrow E_{\hat{u}} \cup \{ (u,v)\in E\mid u\neq \hat{u}\}$
			and $\Gamma'\leftarrow (V,E',w,\langle V_0, V_1\rangle)$.
		\item If $\texttt{contains}(\Gamma', T_{\Gamma})=\texttt{F}$, do the following:
		\begin{itemize}
			\item Compute the least-SEPM ${f'}^*$ of ${\Gamma'}^{w-\nu}$;
			\item If $V_{{f'}^*}=V$:

				-- Push $\hat{u}$ on top of $\texttt{St}$ and $\texttt{insert}(\Gamma', T_{\Gamma})$.

				-- If $\texttt{contains}({f'}^*, T_{\Gamma})=\texttt{F}$,
					then $\texttt{insert}({f'}^*, T_{\Gamma})$ and \texttt{print} ${f'}^*$.
		\end{itemize}
		\end{itemize}
	\end{itemize}
\item While $\texttt{St}\neq\emptyset$:
	\begin{itemize}
	\item \texttt{pop} $\hat{u}$ from $\texttt{St}$;
		Let $E_{\hat{u}}\leftarrow \{(\hat{u},v)\in E\mid f^*(\hat{u}) \prec f^*(v)\ominus(w(\hat{u},v)-\nu)\}$,
		and $E'\leftarrow E_{\hat{u}} \cup \{ (u,v)\in E\mid u\neq \hat{u}\}$,
		and $\Gamma'\leftarrow (V,E',w,\langle V_0, V_1\rangle)$;
	\item Make a recursive call to $\texttt{enum()}$ on input $\Gamma'$.
	\end{itemize}
\end{enumerate}
Down the recursion tree, when computing least-SEPMs, the children Value-Iterations can amortize by
starting from the energy-levels of the parent. The lattice of subgames $\mathcal{B}^*_{\Gamma}$
comprises all and only those subgames $\Gamma'\subseteq \Gamma$ that are eventually inserted
into $T_{\Gamma}$ at Step~(3) of $\texttt{enum}()$; these are called the \emph{basic subgames} of $\Gamma$.
The correctness of $\texttt{enum()}$ follows by Theorem~\ref{thm:main_energystructure}
and Theorem~\ref{Thm:pos_opt_strategy}. In summary, we obtain the following result.
\begin{Thm}\label{thm:listing_algo} There exists a recursive algorithm for enumerating (w/o repetitions)
all elements of $\mathcal{B}^*_{\Gamma}$
with time-delay\footnote{A listing algorithm has $O(f(n))$ \emph{time-delay} when
the time spent between any two consecutives is $O(f(n))$.} $O(|V|^3 |E|\, W)$, on any input \MPG $\Gamma$;
moreover, the algorithm works with $O(|V||E|)+\Theta\big(|E||\mathcal{B}^*_{\Gamma}|)\big)$ space.
So, it enumerates $\mathcal{X}^*_\Gamma$ (w/o repetitions) in $O\big(|V|^3|E|W|\mathcal{B}^*_{\Gamma}|\big)$ total time,
and $O(|V||E|)+\Theta\big(|E||\mathcal{B}^*_{\Gamma}|\big)$ space.
\end{Thm}

To conclude we observe that $\mathcal{B}^*_{\Gamma}$ and $\mathcal{X}^*_\Gamma$ are not isomorphic as lattices,
not even as sets (the cardinality of $\mathcal{B}^*_{\Gamma}$ can be greater that that of $\mathcal{X}^*_\Gamma$).
Indeed, there is a surjective antitone mapping $\varphi_\Gamma$ from $\mathcal{B}^*_{\Gamma}$ onto $\mathcal{X}^*_\Gamma$,
(\ie $\varphi_\Gamma$ sends $\Gamma'\in \mathcal{B}^*_{\Gamma}$ to its least-SEPM $f^*_{\Gamma'}\in \mathcal{X}^*_\Gamma$);
still, we can construct instances of \MPG{s} such that $|\mathcal{B}^*_{\Gamma}| > |\mathcal{X}^*_\Gamma|$,
	\ie $\varphi_\Gamma$ is not into and $\mathcal{B}^*_{\Gamma}$, $\mathcal{X}^*_\Gamma$ are not isomorphic.
That would be a case of \emph{degeneracy}, and an example \MPG $\Gamma_\text{d}$ is given in \figref{fig:ex_degeneracy}.
\begin{figure}[!h] \center
\begin{tikzpicture}[scale=.55, arrows={-triangle 45}, node distance=1.5 and 2]
 		\node[node, thick, color=red, label={below right, xshift=-.6ex, yshift=.6ex:$\sizedcircled{.125ex}{0}$}] (u1) {$u_1$};
		\node[node, thick, color=blue, xshift=0ex, yshift=-5ex, fill=blue!20, above=of u1,
						label={above right, xshift=-.6ex, yshift=-.6ex:$\sizedcircled{.125ex}{0}$}] (u2) {$u_2$};
		\node[node, thick, color=red, xshift=5ex, yshift=5ex, below left=of u1,
						label={left, xshift=.25ex, yshift=0ex:$\sizedcircled{.125ex}{1}$}] (u3) {$u_3$};

		\node[node, thick, color=red,
						label={below left, xshift=.5ex, yshift=.6ex:$\sizedcircled{.125ex}{0}$}, right=of u1] (v1) {$v_1$};
		\node[node, thick, color=blue, xshift=0ex, yshift=-5ex, fill=blue!20, above=of v1,
						label={above left, xshift=.5ex, yshift=-.6ex:$\sizedcircled{.125ex}{0}$}] (v2) {$v_2$};
		\node[node, thick, color=red, xshift=-5ex, yshift=5ex, below right=of v1,
						label={right, xshift=-.25ex, yshift=0ex:$\sizedcircled{.125ex}{1}$}] (v3) {$v_3$};

		\node[node, thick, color=red, xshift=-7ex, yshift=-2.5ex, below right = of u1,
						label={below, xshift=0ex, yshift=.25ex:$\sizedcircled{.125ex}{0}$}] (t) {$t$};

		\node[node, thick, color=red, xshift=-4ex, yshift=4ex, below left=of t,
							label={above left, xshift=.5ex, yshift=-.6ex:$\sizedcircled{.125ex}{0}$}] (u4) {$u_4$};
		\node[node, thick, color=blue, xshift=0ex, yshift=5ex, fill=blue!20, below=of u4,
							label={below right, xshift=-.6ex, yshift=.6ex:$\sizedcircled{.125ex}{0}$}] (u5) {$u_5$};

		\node[node, thick, color=red, xshift=3ex, yshift=4ex, below right=of t,
							label={above right, xshift=-.6ex, yshift=-.6ex:$\sizedcircled{.125ex}{0}$}] (v4) {$v_4$};
		\node[node, thick, color=blue, xshift=0ex, yshift=5ex, fill=blue!20, below=of v4,
							label={below left, xshift=.5ex, yshift=.6ex:$\sizedcircled{.125ex}{0}$}] (v5) {$v_5$};
		\draw[color=red] (u1) to [bend right=30] node[right] {$0$} (u2);
		\draw[] (u2) to [bend right=30] node[left] {$0$} (u1);
		\draw[] (u3) to [] node[above left, yshift=-1ex] {$-2$} (u1);
		\draw[color=red] (u3) to [] node[above] {$-1$} (t);
		\draw[] (t) to [] node[above] {$-10$} (u4);
		\draw[color=red] (u4) to [bend right=30] node[left] {$0$} (u5);
		\draw[] (u5) to [bend right=30] node[right] {$0$} (u4);

		\draw[] (v1) to [bend right=30] node[right] {$0$} (v2);
		\draw[] (v2) to [bend right=30] node[left] {$0$} (v1);
		\draw[] (v3) to [] node[above right, yshift=-1ex] {$-2$} (v1);
		\draw[color=red] (v3) to [] node[above] {$-1$} (t);
		\draw[color=red] (t) to [] node[above] {$0$} (v4);
		\draw[color=red] (v4) to [bend right=30] node[left] {$0$} (v5);
		\draw[] (v5) to [bend right=30] node[right] {$0$} (v4);
\end{tikzpicture}
\caption{An MPG $\Gamma_\text{d}$ for which $|\mathcal{B}^*_{\Gamma}| > |\mathcal{X}^*_\Gamma|$.}\label{fig:ex_degeneracy}
\end{figure}
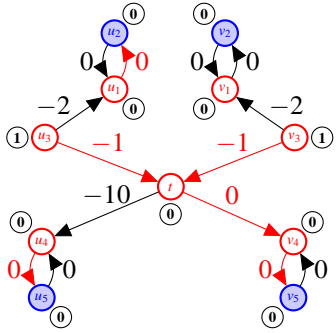

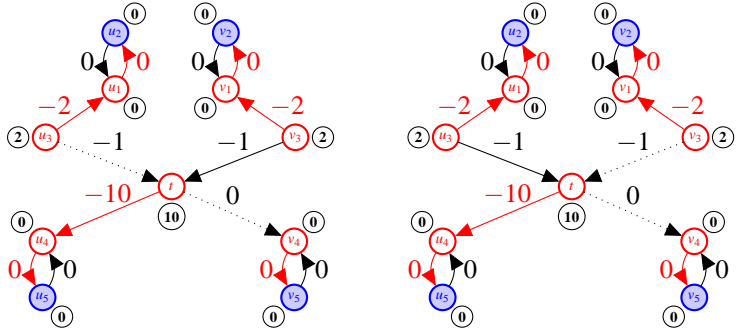
\begin{figure}[!h] \center
\begin{tikzpicture}[scale=.55, arrows={-triangle 45}, node distance=1.5 and 2]
 		\node[node, thick, color=red, label={below right, xshift=-.6ex, yshift=.6ex:$\sizedcircled{.125ex}{0}$}] (u1) {$u_1$};
		\node[node, thick, color=blue, xshift=0ex, yshift=-5ex, fill=blue!20, above=of u1,
							label={above right, xshift=-.6ex, yshift=-.6ex:$\sizedcircled{.125ex}{0}$}] (u2) {$u_2$};
		\node[node, thick, color=red, xshift=5ex, yshift=5ex, below left=of u1,
							label={left, xshift=.25ex, yshift=0ex:$\sizedcircled{.125ex}{2}$}] (u3) {$u_3$};

		\node[node, thick, color=red, label={below left, xshift=.5ex, yshift=.6ex:$\sizedcircled{.125ex}{0}$}, right=of u1] (v1) {$v_1$};
		\node[node, thick, color=blue, xshift=0ex, yshift=-5ex, fill=blue!20, above=of v1,
							 label={above left, xshift=.5ex, yshift=-.6ex:$\sizedcircled{.125ex}{0}$}] (v2) {$v_2$};
		\node[node, thick, color=red, xshift=-5ex, yshift=5ex, below right=of v1,
								label={right, xshift=-.25ex, yshift=0ex:$\sizedcircled{.125ex}{2}$}] (v3) {$v_3$};

		\node[node, thick, color=red, xshift=-7ex, yshift=-2.5ex, below right = of u1,
								label={below, xshift=0ex, yshift=.25ex:$\sizedcircled{.125ex}{10}$}] (t) {$t$};

		\node[node, thick, color=red, xshift=-4ex, yshift=4ex, below left=of t,
						label={above left, xshift=.5ex, yshift=-.6ex:$\sizedcircled{.125ex}{0}$}] (u4) {$u_4$};
		\node[node, thick, color=blue, xshift=0ex, yshift=5ex, fill=blue!20, below=of u4,
						label={below right, xshift=-.6ex, yshift=.6ex:$\sizedcircled{.125ex}{0}$}] (u5) {$u_5$};

		\node[node, thick, color=red, xshift=3ex, yshift=4ex, below right=of t,
							label={above right, xshift=-.6ex, yshift=-.6ex:$\sizedcircled{.125ex}{0}$}] (v4) {$v_4$};
		\node[node, thick, color=blue, xshift=0ex, yshift=5ex, fill=blue!20, below=of v4,
							label={below left, xshift=.5ex, yshift=.6ex:$\sizedcircled{.125ex}{0}$}] (v5) {$v_5$};
		\draw[color=red] (u1) to [bend right=30] node[right] {$0$} (u2);
		\draw[] (u2) to [bend right=30] node[left] {$0$} (u1);
		\draw[color=red] (u3) to [] node[above left, yshift=-1ex] {$-2$} (u1);
		\draw[dotted] (u3) to [] node[above] {$-1$} (t);
		\draw[color=red] (t) to [] node[above] {$-10$} (u4);
		\draw[color=red] (u4) to [bend right=30] node[left] {$0$} (u5);
		\draw[] (u5) to [bend right=30] node[right] {$0$} (u4);

		\draw[color=red] (v1) to [bend right=30] node[right] {$0$} (v2);
		\draw[] (v2) to [bend right=30] node[left] {$0$} (v1);
		\draw[color=red] (v3) to [] node[above right, yshift=-1ex] {$-2$} (v1);
		\draw[] (v3) to [] node[above] {$-1$} (t);
		\draw[dotted] (t) to [] node[above] {$0$} (v4);
		\draw[color=red] (v4) to [bend right=30] node[left] {$0$} (v5);
		\draw[] (v5) to [bend right=30] node[right] {$0$} (v4);
\end{tikzpicture}
\qquad
\begin{tikzpicture}[scale=.55, arrows={-triangle 45}, node distance=1.5 and 2]
 		\node[node, thick, color=red,
						label={below right, xshift=-.6ex, yshift=.6ex:$\sizedcircled{.125ex}{0}$}] (u1) {$u_1$};
		\node[node, thick, color=blue, xshift=0ex, yshift=-5ex, fill=blue!20, above=of u1,
						label={above right, xshift=-.6ex, yshift=-.6ex:$\sizedcircled{.125ex}{0}$}] (u2) {$u_2$};
		\node[node, thick, color=red, xshift=5ex, yshift=5ex, below left=of u1,
						label={left, xshift=.25ex, yshift=0ex:$\sizedcircled{.125ex}{2}$}] (u3) {$u_3$};
		\node[node, thick, color=red,
							label={below left, xshift=.5ex, yshift=.6ex:$\sizedcircled{.125ex}{0}$}, right=of u1] (v1) {$v_1$};
		\node[node, thick, color=blue, xshift=0ex, yshift=-5ex, fill=blue!20, above=of v1,
							label={above left, xshift=.5ex, yshift=-.6ex:$\sizedcircled{.125ex}{0}$}] (v2) {$v_2$};
		\node[node, thick, color=red, xshift=-5ex, yshift=5ex, below right=of v1,
							label={right, xshift=-.25ex, yshift=0ex:$\sizedcircled{.125ex}{2}$}] (v3) {$v_3$};

		\node[node, thick, color=red, xshift=-7ex, yshift=-2.5ex, below right = of u1,
							label={below, xshift=0ex, yshift=.25ex:$\sizedcircled{.125ex}{10}$}] (t) {$t$};

		\node[node, thick, color=red, xshift=-4ex, yshift=4ex, below left=of t,
							label={above left, xshift=.5ex, yshift=-.6ex:$\sizedcircled{.125ex}{0}$}] (u4) {$u_4$};
		\node[node, thick, color=blue, xshift=0ex, yshift=5ex, fill=blue!20, below=of u4,
							label={below right, xshift=-.6ex, yshift=.6ex:$\sizedcircled{.125ex}{0}$}] (u5) {$u_5$};
		\node[node, thick, color=red, xshift=3ex, yshift=4ex, below right=of t,
							label={above right, xshift=-.6ex, yshift=-.6ex:$\sizedcircled{.125ex}{0}$}] (v4) {$v_4$};
		\node[node, thick, color=blue, xshift=0ex, yshift=5ex, fill=blue!20, below=of v4,
							label={below left, xshift=.5ex, yshift=.6ex:$\sizedcircled{.125ex}{0}$}] (v5) {$v_5$};
		\draw[color=red] (u1) to [bend right=30] node[right] {$0$} (u2);
		\draw[] (u2) to [bend right=30] node[left] {$0$} (u1);
		\draw[color=red] (u3) to [] node[above left, yshift=-1ex] {$-2$} (u1);
		\draw[] (u3) to [] node[above] {$-1$} (t);
		\draw[color=red] (t) to [] node[above] {$-10$} (u4);
		\draw[color=red] (u4) to [bend right=30] node[left] {$0$} (u5);
		\draw[] (u5) to [bend right=30] node[right] {$0$} (u4);

		\draw[color=red] (v1) to [bend right=30] node[right] {$0$} (v2);
		\draw[] (v2) to [bend right=30] node[left] {$0$} (v1);
		\draw[color=red] (v3) to [] node[above right, yshift=-1ex] {$-2$} (v1);
		\draw[dotted] (v3) to [] node[above] {$-1$} (t);
		\draw[dotted] (t) to [] node[above] {$0$} (v4);
		\draw[color=red] (v4) to [bend right=30] node[left] {$0$} (v5);
		\draw[] (v5) to [bend right=30] node[right] {$0$} (v4);
\end{tikzpicture}
\caption{Two basic subgames $\Gamma^1_d\neq \Gamma^2_d$ of $\Gamma_\text{d}$,
						having the same least-SEPM $f^*_1=f^*_2$.}\label{fig:ex_degeneracy_subgames}
\end{figure}
In the MPG $\Gamma_\text{d}$, Player~0 has to decide how to move only at $u_3,v_3$ and $t$; the remaining moves are forced.
	The least-SEPM $f^*$ of $\Gamma_\text{d}$ is:
		$f^*(u_3)=1$, $f^*(v_3)=1$, $f^*(t)=0$, and $\forall_{x\in V_{\Gamma_\text{d}}\setminus\{u_3,v_3,t\}}\, f^*(x)=0$;
			leading to the following memory-less strategy: $\sigma^*_0(u_3)=t$, $\sigma^*_0(v_3)=t$, $\sigma^*_0(t)=v_4$.
Then, consider the lattice of subgames $\mathcal{B}^*_{\Gamma_\text{d}}$;
	particularly, consider the following two basic subgames $\Gamma^1_d$, $\Gamma^2_d$:
	let $\Gamma'_d$ be the arena obtained by removing the arc $(t,v_4)$ from $\Gamma_d$;
	let $\Gamma^1_d$ be the arena obtained by removing the arc $(u_3,t)$ from $\Gamma'_d$;
	let $\Gamma^2_d$ be the arena obtained by removing the arc $(v_3,t)$ from $\Gamma'_d$.
See \figref{fig:ex_degeneracy_subgames} for an illustration.
	Next, let $f^*_1,f^*_2$ be the least-SEPMs of $\Gamma^1_d$ and $\Gamma^2_d$, respectively;
	 then, $f^*_1(u_3)=f^*_2(u_3)=2$, $f^*_1(v_3)=f^*_2(v_3)=2$, $f^*_1(t)=f^*_2(t)=10$,
		and $\forall_{x\in V_{\Gamma_\text{d}}\setminus\{u_3,v_3,t\}}\, f^*_1(x)=f^*_2(x)=0$.
Thus, $\Gamma^1_d\neq \Gamma^2_d$, but $f^*_1=f^*_2$; this proves that $\Gamma_\text{d}$ is \emph{degenerate}
	and that $\mathcal{B}^*_{\Gamma}$, $\mathcal{X}^*_\Gamma$ are not isomorphic.